\documentclass[journal, twocolumn]{IEEEtran}
\usepackage[english]{babel}
\usepackage{amsmath,amssymb,bm}
\usepackage{amsfonts}
\usepackage{graphicx}
\usepackage{subfigure}
\usepackage{epstopdf}
\usepackage{color}

\newcommand{\qb}{{\bf b}}

\newcommand{\qe}{{\bf e}}

\newcommand{\qh}{{\bf h}}

\newcommand{\qw}{{\bf w}}

\newcommand{\qC}{{\bf C}}

\newcommand{\qG}{{\bf G}}
\newcommand{\qH}{{\bf H}}
\newcommand{\qI}{{\bf I}}

\newcommand{\qW}{{\bf W}}

\newcommand{\qzero}{{\bf 0}}

\newcommand{\tr}{\mbox{trace}}
 
\newtheorem{theorem}{Theorem}
\newcommand{\be}{\begin{equation}} \newcommand{\ee}{\end{equation}}
\newcommand{\bea}{\begin{eqnarray}} \newcommand{\eea}{\end{eqnarray}}

 \def\endproof{\hspace*{\fill}~$\square$\par\endtrivlist\unskip}

\usepackage[paper=letterpaper,
            marginparwidth=0.0in,     
            marginparsep=.05in,       
            margin=0.8in,               
            includemp]{geometry}
\begin{document}

\title{Deep Learning Enabled Optimization of Downlink Beamforming Under Per-Antenna Power Constraints: Algorithms and Experimental Demonstration }%
\author{Juping Zhang,
        Wenchao Xia, ~\IEEEmembership{Student~Member,~IEEE,}
        Minglei You,
        Gan Zheng,~\IEEEmembership{Senior~Member,~IEEE,}
        Sangarapillai Lambotharan,~\IEEEmembership{Senior~Member,~IEEE,}
        and Kai-Kit Wong,~\IEEEmembership{Fellow,~IEEE}
     \thanks{J. Zhang,  G. Zheng,  and S. Lambotharan  are  with the Wolfson School of Mechanical, Electrical and Manufacturing Engineering, Loughborough University, Leicestershire, LE11 3TU, UK (Email: \{j.zhang3,g.zheng, s.lambotharan\}@lboro.ac.uk).}
     \thanks{M. You is  with Department of Engineering, Durham University, Durham, DH1 3LE, UK (Email: minglei.you@durham.ac.uk).}
  \thanks{W. Xia is with the Information Systems Technology and Design Pillar, Singapore University of Technology and Design, Singapore 487372. (E-mail: wenchao xia@sutd.edu.sg).}
  \thanks{K.-K. Wong is with the Department of Electronic and Electrical Engineering, University College London, London, WC1E 6BT, UK (Email: kai-kit.wong@ucl.ac.uk).}
  \thanks{We  gratefully acknowledge the support of NVIDIA Corporation with the donation of a  Titan Xp GPU used for this research.}
}
\maketitle

\begin{abstract}
This paper studies fast downlink beamforming algorithms using deep learning in  multiuser multiple-input-single-output  systems where each transmit antenna at the base station has its own power constraint. We focus on the signal-to-interference-plus-noise ratio (SINR) balancing problem which is quasi-convex but there is no efficient solution available. We first design a fast subgradient algorithm that can achieve near-optimal solution with reduced complexity. We then propose a deep neural network structure to learn the optimal beamforming based on convolutional networks and  exploitation of the duality of the original problem. Two strategies of learning various dual variables are investigated with different accuracies, and the corresponding recovery of the original  solution is facilitated by the subgradient algorithm. We also develop a generalization method of the proposed algorithms so that they can adapt to the varying number of   users and antennas without re-training.  We carry out intensive numerical simulations and testbed experiments to evaluate the performance of the proposed algorithms.
{Results show that the proposed algorithms achieve close to optimal solution in simulations with perfect channel  information  and outperform the alleged {{theoretically}} optimal solution in experiments, illustrating a  better performance-complexity tradeoff than existing schemes.}
 \end{abstract}


\section{Introduction}
Multiuser multi-antenna techniques  (or multiple-input multiple-output, MIMO) techniques  can significantly improve the spectral and energy efficiency  of wireless communications   by exploiting the degree of freedom in the spatial domain. They have been widely adopted in modern wireless communications systems such as the fourth and the fifth-generation (4G and 5G) of cellular networks \cite{MUMIMO-4G}\cite{5G-2014}, the high efficiency wireless local area  (WiFi) networks standard 802.11ax \cite{MUMIMO-wifi}, and the latest  satellite digital video broadcasting  standard  DVB-S2X \cite{dvb-s2}. Among the multiuser MIMO techniques, beamforming is one of the most promising and practical schemes to mitigate multiuser interference and exploit the gain of MIMO antennas.

In the last two decades, the optimal beamforming strategies have been intensively studied for the multiple-input single-output (MISO) downlink where a base station with multiple antennas serves multiple single-antenna users. For instance, the problem of  signal-to-interference-plus-noise ratio  (SINR) balancing  or maximization of the minimum SINR of all users,  under a total power constraint was studied in  \cite{schubert2004solution,bjornson2014optimal},  the total BS transmit power minimization problem under quality of service (QoS) constraints was investigated in \cite{rashid1998transmit,wiesel2006linear,gershman2010convex,shi2016sinr},  and the sum rate maximization problem under the total   power constraint was tackled in \cite{bjornson2014optimal,Yoo2006on, christensen2008weighted,shi2011an}. The existing approaches mainly make use of the advances of convex optimization techniques such as second-order cone programming (SOCP) \cite{wiesel2006linear, gershman2010convex} and semidefinite programming (SDP)  \cite{luo2010semidefinite}, and the uplink-downlink duality which indicates that under the sum power constraint, the achievable SINR region and the normalized beamforming in the downlink are the same as those in the dual uplink channel.

 Early works   mostly focus on the optimal beamforming design under the sum power constraint  across all antennas of a transmitter. This constraint does not take into account the fact that each transmit antenna has its own power amplifier, and therefore its power is individually limited. The per-antenna power constraints were first systematically studied in \cite{yu2007transmitter} where a dual framework was proposed to minimize the maximum transmit power of each antenna under users' SINR constraints. This work has sparked much research interest in optimizing  beamforming under per-antenna power constraints. The work in \cite{ZFDPC-PA} studied the optimization of the nonlinear zero forcing (ZF) dirty paper coding based beamforming under per-antenna power constraints. Generic optimization of beamforming for  multibeam satellite systems was studied in \cite{Zheng-12} under general linear and nonlinear power constraints. The per-antenna constant envelope precoding for large multiuser MIMO systems was investigated in \cite{Larsson-13}. The transceiver designs for multi-antenna multi-hop cooperative communications under per-antenna power constraints were proposed in \cite{Xing-PA} and both linear  and nonlinear transceivers  were investigated.  The signal-to-leakage-plus-noise ratio (SLNR)  maximized precoding for the  downlink under per-antenna power constraints  was considered in \cite{Shen-16} where a semi-closed form optimal solution was proposed.
A general framework for covariance matrix optimization of MIMO systems under different types of   power constraints was proposed in \cite{xing-17}. More recently, the optimal MIMO precoding under the constraints of both the total   consumed power constraint and the individual radiated power constraints was studied in \cite{Larsson-19} and numerical algorithms were developed to maximize the mutual information.

The problem of interest in this paper is to efficiently maximize the minimum received SINR  or to balance SINR, in the multiuser MISO downlink under per-antenna power constraints at the BS. This problem, although being quasiconvex, is more challenging than the counterpart with the total power constraint and the problem of minimizing the per-antenna power in \cite{yu2007transmitter}, and until now   there does not exist   efficient algorithms.  Consequently,   existing beamforming techniques are unable to support real-time applications because the small-scale fading channel varies considerably  fast.
 {{For instance, in a WLAN 802.11n system operating at 2.4 GHz with a  pedestrian speed of $1.4$ m/s,  the coherence time is $89$ ms; and in a  Long-Term Evolution (LTE) downlink operating at 2.6 GHz with a residential area vehicle velocity of $10$ m/s,  the coherence time is only $11.5$ ms. Traditional time-consuming optimization routines will produce obsolete beamforming solution that is not timely for the current channel state and lead to significant performance degradation which will be demonstrated in our experiment.}}
 In \cite{Fettweis-08}, the dual problem  was   derived and the optimal solution  at much reduced computational cost was developed. However, it was found out that the best solution is obtained by a commercial nonlinear solver \cite{knitro}, which does not explore the structure of the problem and is still not efficient. Although there are simple heuristic beamforming solutions which have closed-form solutions such as  the ZF beamforming and the regularized ZF (RZF) beamforming, the reduced complexity often leads to performance loss. Even worse, the work in  \cite{ZF-PA} showed that the conventional ZF beamforming under per-antenna power constraints no longer admits  a simple pseudo-inverse form as the case under the total power constraint, and instead the optimal ZF beamforming requires solving an SOCP problem which has much higher complexity.

In this paper, we take a different approach and develop deep learning (DL) enabled  beamforming solutions  to dramatically improve the computational efficiency.
 Recently DL has been recognized as a promising solution for addressing various problems in several areas of wireless networks. This is because deep neural networks have the ability to model highly non-linear functions at considerably low complexity. One of the areas of interest is to deal with   scenarios in which the channel model does not exist, e.g., in underwater and molecular communications \cite{farsad2017detection} or is difficult to characterize analytically due to imperfections and nonlinearities  \cite{Shafin-review}. In these situations,
DL based detection has been proposed to  tackle the underlying unknown nonlinearities \cite{Mosleh-18}.  Another area of interest  is to optimize the end-to-end system performance \cite{oshea2017learning,Dorner2018deep}. Conventional communication systems are based on the modular design and each block (e.g., coding, modulation) is optimized independently, which can not guarantee  the optimal overall performance.  However, DL holds great promises for further improvement by considering  end-to-end performance optimization. The third area of interest is to overcome the complexity of wireless networks  \cite{Shafin-review} which is  the focus of our paper.
In this aspect, DL has found many exciting applications in wireless communications such as channel decoding \cite{liang2018an,kim2018communication}, MIMO detection \cite{he2018model,samuel2018learning}, channel estimation \cite{he2018deep,wen2018deep}.
 The current work belongs to the framework of learning to optimize in wireless resource allocation. {{The rationale is that the DL technique bypasses the complex optimization procedures, and  learns  the optimal mapping from the channel state to produce the beamforming solution directly by training a neural network. The result is that the trained neural network can be used as a function mapping to obtain the real-time beamforming solution  with channel state as input.}} As a result, the computational complexity is   transferred to   offline training phase\footnote{To the best of our knowledge, the computational complexity of the training phase is not well understood, due to the complex implementation of the backpropagation process and that it depends very much on the specific application regarding the required number of training examples for satisfactory generalization. That said, this is usually not a concern in most applications because training takes place offline given sufficient computational capability and retraining is only performed infrequently when the specific applications depart considerably from those training examples.}, and hence the complexity during the online transmission phase is greatly reduced.  The mostly successful applications of DL in this framework by far is power allocation \cite{sun2017learning,liang2018towards,lee2018deep,li2018intellignet,Emil-19}, in which the power vector is treated as the training output, while the channel gains are taken into the input of the DL network. In this case, the power variables only take positive values and the number of power variables is normally the number of users and therefore relatively small and easy to handle.

However,  there are few works that focus on the learning approach to optimize the beamforming design in multi-antenna  communications, with the exception of \cite{shi2018learning,Kerret2018robust,alkhateeb2018deep,huang2018unsupervised,Wenchao-ICC,Wenchao-2019}. The difficulty is partly due to the large number of complex variables contained in the beamforming matrix that need to be optimized.
An outage-based approach to transmit beamforming was studied in \cite{shi2018learning} to deal with the channel uncertainty at the BS, however, only a single user was considered. The work in \cite{Kerret2018robust}  designed a decentralized robust precoding scheme based on a deep neural network (DNN). {{The projection over a finite dimensional subspace in \cite{Kerret2018robust} reduced the difficulty, but also limited the performance. A DL model was used in \cite{alkhateeb2018deep} to predict the beamforming matrix directly from the signals received at the distributed BSs in millimeter wave systems.
    However, both  \cite{Kerret2018robust} and  \cite{alkhateeb2018deep}  predicted the beamforming matrix in the finite solution space at the cost of performance loss. The works in \cite{shi2018learning,huang2018unsupervised} directly estimated the beamforming matrix without exploiting the problem structure in which the number of variables to predict increases significantly as the numbers of transmit antennas  and users increase.  This will lead to high training complexity and low learning accuracy of the neural networks when the numbers of transmit antennas  and users are large. In our previous works \cite{Wenchao-ICC}\cite{Wenchao-2019}, we proposed a beamforming neural network to optimize the beamforming vectors, but it is restricted to the total power constraint. We notice that none of existing works addressed the SINR balancing problem under the practical per-antenna power constraints, for which DL solution becomes even more attractive.

In this paper, we propose a DL enabled beamforming optimization approach for SINR balancing to provide an improved performance-complexity tradeoff under per-antenna power constraints. Inspired by the model driven learning philosophy \cite{he2018model2}, we propose to first learn the dual variables with reduced dimension rather than the original large beamforming matrix and then  recover the beamforming solution from the learned dual solution, by exploiting the structure or model of the beamforming optimization problem. Our main contributions are summarized as follows:
\begin{itemize}
  \item A subgradient algorithm is first proposed which not only demonstrates faster convergence than the best known algorithm in \cite{Fettweis-08}, but also facilitates the development of the DL solutions.
  \item A general DL structure to learn the dual variables is proposed, and two learning strategies are proposed to achieve the performance-complexity tradeoff.
      A heuristic method is developed to facilitate the generalization of the proposed DL algorithms by augmenting the training set  so that they can adapt to the varying number of active  users and antennas without re-training.
  \item Both software simulations and testbed experiments using software defined radio (SDR) are carried out to validate the performance of the proposed algorithms. To the best of our knowledge, this is the first testbed demonstration of deep learning enabled multiuser beamforming.

\end{itemize}

The remainder of this paper is organized as follows. Section \ref{section system model} introduces the system model and formulates the SINR balancing problem and its dual formulation.  Section \ref{subgradient} proposes the subgradient algorithm.   Section \ref{DL} provides the general structure framework for the beamforming optimization based on learning the dual variables and the recovery algorithms. Numerical and experimental results are presented in Section \ref{Performance}. Finally, conclusion is drawn in Section \ref{conclusions}.

\textbf{Notations:} The notations are given as follows. Matrices and vectors are denoted by bold capital and lowercase symbols, respectively. $(\cdot)^T$, $(\cdot)^*$, $(\cdot)^\dag$ and $(\cdot)^{-1}$  stand for transpose, conjugate, conjugate transpose and inverse/pseudo inverse (when applicable) operations of a matrix, respectively.
  $\mathbf{A}\succ \qzero$ indicates that the matrix $\mathbf{A}$ is positive definite. The operator $\text{diag}(\mathbf{a})$ denotes the operation to diagonalize the vector $\mathbf{a}$ into a matrix whose main diagonal elements are from $\mathbf{a}$.   Finally, $\mathbf{a}\sim\mathcal{CN}(\mathbf{0},\bm{\Sigma})$ represents a complex Gaussian vector with zero-mean and covariance matrix $\bm{\Sigma}$. ${\mathbb Z}$ denotes the non-negative field.

\section{System Model and Problem Formulation}\label{section system model} 
 Consider an MISO downlink channel where an $N_t$-antenna BS
 transmits   signals  to $K$ single-antenna users. For the user $k$, its channel vector,
 beamforming vector,  and data symbol   are
 denoted as $\qh_k^T, \qw_k$,    $s_k$, respectively, where ${\tt} E(|s_k|^2)=1$.
The additive white Gaussian noise (AWGN) at the received is denoted as $n_k\sim \mathcal{CN}(0, N_0)$.
All wireless links exhibit independent frequency non-selective
Rayleigh block  fading.     The received signal at user $k$ is
 \bea\label{eqn:rs}
    y_k &=& {\qh_k^T} \sum_{i=1}^K \qw_i s_i + n_k.
 \eea
 The SINR at the receiver of user $k$ is given by
\begin{equation}\label{eqn:SINR}
{\gamma}_k= \frac{ |\qh_{k}^T\qw_k|^2}{\sum\limits_{i=1,i \ne
k}^K |\qh_{k}^T\qw_i|^2 + N_0}.
\end{equation}

The beamforming matrix is collected in $\qW=[\qw_1, \qw_2, \cdots, \qw_K] \in {\mathbb C}^{N_t\times K}$. Then the per-antenna power   at antenna $n$ can be expressed as
\be
    p_n= \|\mathbf{W}(n,:)\|^2 = \|\qe_n^T \qW\|^2,
\ee
where $\qe_n$ is a zero vector except its $n-$th element being 1.

The problem  of interest is to maximize the minimum user SINR, i.e., SINR balancing, under per-antenna power constraints $\{P_n\}$. Mathematically, it can be formulated as follows:
\bea
    \textbf{P1:}&& \max_{\qW, \Gamma} ~~\Gamma \notag\\
    \mbox{s.t.}  && {\gamma}_k= \frac{ |\qh_{k}^T\qw_k|^2}{\sum\limits_{i=1,i \ne
k}^K |\qh_{k}^T\qw_i|^2 + N_0} \ge \Gamma, \forall k, \label{P1:SINR}\\
     &&p_n= \|\qe_n^T \qW\|^2 \le P_n, ~~~\forall n.\label{P1:PA}
\eea

The SINR balancing problem is in general quasi-convex, so it can be solved via methods such as bisection search and generalized eigenvalue programming \cite{wiesel2006linear}\cite{Fettweis-08}. However, these methods suffer from high  complexity and computational delay, and are not practical for real-time data transmissions.

In \cite{Fettweis-08}, a useful dual formulation of \textbf{P1}  is derived as
\bea
     \textbf{P2:}&& \max_{\beta, \bm \lambda, \bm \mu} ~~ \beta \notag\\
    \mbox{s.t.} && \beta \lambda_k \qh_k^T \qG(\bm\lambda, \bm \mu)^{-1}\qh_k^*\le 1, \forall k,\\
                && \sum_{k=1}^K \lambda_k N_0 = 1,\notag\\
                && \sum_{n=1}^{N_t} \mu_n P_n=1,   \notag\\
                && \bm \lambda, \bm \mu, \beta\ge \qzero.
\eea
 where $\qG(\bm\lambda, \bm \mu) \triangleq \sum_{i=1}^K \lambda_i \qh_i^*\qh_i^T +\mbox{Diag}(\bm \mu)$,   $\bm \lambda \in {\mathbb Z^{K}}, \bm \mu\in {\mathbb Z^{N_t}}$ are dual variables associated with the SINR constraint \eqref{P1:SINR} and the per-antenna power constraint \eqref{P1:PA} in \textbf{P1}, and $\beta$ is related to the minimum  SINR  $\Gamma$ in \textbf{P1} by the relation $\beta = 1 + \frac{1}{\Gamma}$. {{For the solution of \textbf{P2}, it is assumed that  $\qG(\bm\lambda, \bm \mu)\succ\qzero$.}}

Although the problem \textbf{P2} is still a quasi-convex problem, compared to the original problem  \textbf{P1}, it can be more efficiently solved  because it only involves $K+N_t+1$ non-negative variables while \textbf{P1} needs to optimize $2KN_t$ real variables. The problem \textbf{P2} can be solved using standard nonlinear solvers such as  Matlab's built-in function `fmincon'. Currently the fastest optimal solution is known to be achieved by Ziena's nonlinear solver Knitro \cite{knitro}, which is compared and shown in \cite{Fettweis-08}. However, the general solvers do not exploit the special analytical properties of the problem  \textbf{P2}, so they are not efficient. In addition, it is not known how to recover the optimal solution to the beamforming matrix $\qW^*$ once \textbf{P2} is solved. These issues will be studied in the next section.


\section{A Subgradient Algorithm to Solve \textbf{P2}}  \label{subgradient} 

In this section, we   derive a fast subgradient algorithm to solve \textbf{P2}, based on the downlink-uplink duality results derived in \cite{yu2007transmitter}. According to \cite[Theorem 1]{yu2007transmitter}, the problem \textbf{P2} can be equivalently written as the following max-max problem:
\bea
     \textbf{P3:}&& \max_{\bm \mu} \max_{\Gamma, \bm \lambda} ~~ \Gamma \notag\\
    \mbox{s.t.} &&{{\max_{\qw_k}}}\frac{\lambda_k |\bar\qw_k^\dag\qh_k^*|^2}{ \sum_{i=1,i\ne k}^K  {{\lambda_i}}|\bar\qw_k^\dag\qh_i^*|^2 + \bar\qw_k^\dag(\mbox{Diag}(\bm \mu)) \bar\qw_k}\ge \Gamma, \forall k, \notag \\ 
                && \sum_{k=1}^K \lambda_k = \frac{1}{N_0},\notag\\
                && \sum_{n=1}^{N_t} \mu_k P_n=1,\notag\\
                && \bm \lambda, \bm \mu, \Gamma\ge \qzero.
\eea
\textbf{P3} can be interpreted as the maximization of the minimum user SINR in the virtual uplink in which $K$ single-antenna users transmit signals to the BS with the total power constraint $\frac{1}{N_0}$. The uncertain covariance matrix of the received noise vector is characterized by $\mbox{Diag}(\bm \mu)$.  {{The normalized  receive beamforming at the BS for user $k$ is denoted by $\bar \qw_k=\frac{\qw_k}{\|\qw_k\|}$}}  which has the same direction as the downlink transmit beamforming, while $\lambda_k$ denotes the uplink transmit power of user $k$. Because the covariance matrix of the received noise vector $\mbox{Diag}(\bm \mu)$ is also a variable, \textbf{P3} is still difficult to solve. To tackle this problem, we first keep the  variable $\bm \mu$ fixed, and then reach the sub-problem below:
\bea
     \textbf{P4:}&&  f(\bm\mu)=\max_{\Gamma, \bm \lambda} ~~ \Gamma \notag\\
    \mbox{s.t.} &&{{\max_{\qw_k}}}\frac{\lambda_k |\bar\qw_k^\dag\qh_k^*|^2}{ \sum_{i=1,i\ne k}^K  {{\lambda_i}} |\bar\qw_k^\dag\qh_i^*|^2 + \bar\qw_k^\dag(\mbox{Diag}(\bm \mu)) \bar\qw_k}\ge \Gamma, \forall k, \label{eqn:P41} \\ 
                && \sum_{k=1}^K \lambda_k N_0 = 1,\label{eqn:P42}\\
                && \bm \lambda,   \Gamma\ge \qzero.
\eea
  \textbf{P4} can be interpreted as the  nonlinear SINR balancing problem with a total power constraint and colored noise with covariance matrix $\mbox{Diag}(\bm \mu)$. In the following, we propose an efficient fixed-point iteration in Algorithm 1 below  to solve \textbf{P4}. 

\textbf{\underline{Algorithm 1 to Solve \textbf{P4}:}}
\begin{enumerate}
    \item Initialize $\bm \lambda$ that satisfies $\sum_{k=1}^K  \lambda_k N_0=1$. Suppose $j$ is the iteration index, and the achievable SINR in the uplink is $\gamma^{(j)}$. Repeat the following   steps 2)-5) until convergence.
    \item For each $k$, define $\qG_k(\bm\lambda, \bm \mu) \triangleq \sum_{i=1, i\ne k}^K \lambda_i \qh_i^*\qh_i^T +\mbox{Diag}(\bm \mu)$.

    \item Solve an auxiliary variable $\bar \lambda_k$ as
        \be
            \bar \lambda_k = \mathcal{I}_k (\bm\lambda^{(j-1)} )\triangleq  \gamma^{(j)} \frac{1}{\qh_k^T \qG_k(\bm\lambda, \bm \mu)^{-1}\qh_k^*}, \forall k.
        \ee
    \item Normalize $\{\bar \lambda_k \}$ to obtain $\{\lambda_k\}$ as:
        \be
            \lambda_k =  \bar \lambda_k \eta, \mbox{where}~~ \eta = \frac{1}{\sum_{i=1}^K \bar \lambda_i N_0}.
        \ee
        \item
          Calculate $\beta_k = \frac{1} {\lambda_k \qh_k^T \qG(\bm\lambda, \bm \mu)^{-1}\qh_k^*}$. Then update the achievable SINR in the uplink as
          \be
            \gamma^{(j)} = \min_k  \frac{1}{\beta_k-1}.
          \ee
 \end{enumerate}
 It can be proved that Algorithm 1 converges to the optimal solution of \textbf{P4}. The proof  is similar to \cite[Theorem 11.1]{proof} and a refined version is provided in Appendix A for completeness. 

 The optimal uplink beamforming for a given $\bm\mu$ can be derived according to the minimum mean square error (MMSE) criterion:
 \be
    \bar \qw_k = \frac{\qG_k(\bm\lambda, \bm \mu)^{-1}\qh_k^*}{\|\qG_k(\bm\lambda, \bm \mu)^{-1}\qh_k^*\|}, \forall k.
 \ee

 With the inner maximization problem \textbf{P4} solved for given $\bm \mu$, we can obtain the objective function value $f(\bm\mu)$. Next we solve the outer maximization of $\bm \mu$ using a subgradient projection algorithm, where the subgradient can be found using the downlink beamforming obtained from the normalized uplink beamforming.

 {{As proved in Appendix B,  $f(\bm\mu)$ is a concave function  in  $\bm \mu$.}} A subgradient of $\mu_n$ can be expressed as $\|\qe_n^T \qW\|^2$ because $\mu_n$ is the dual variable associated with the $n$-th  antenna power constraint.
 The proof is omitted.  Based on this result, we propose the following subgradient based algorithm \cite{yu2007transmitter}\cite{subgradient} to solve  \textbf{P3}.

\textbf{\underline{Algorithm 2 to Solve \textbf{P3}:}}
 \begin{enumerate}
    \item Initialize $\bm\mu$. Suppose $j$ is the iteration index.  Repeat the following  steps 2)-7) until convergence.

    \item Given $\bm\mu$, call Algorithm 1 to find the optimal $\bar{\bm\lambda}$.

    \item
          Calculate $\beta_k = \frac{1} {\lambda_k \qh_k^T \qG(\bar{\bm\lambda}, \bm \mu)^{-1}\qh_k^*}$. Then update the achievable SINR in the uplink as
          \be
            \gamma^{(j)} = \min_k  \frac{1}{\beta_k-1}.
          \ee

    \item Find the optimal normalized uplink beamforming
           \be
           \bar \qw_k = \frac{\qG_k(\bar{\bm\lambda}, \bm \mu)^{-1}\qh_k^*}{\|\qG_k(\bar{\bm\lambda}, \bm \mu)^{-1}\qh_k^*\|}, \forall k.
         \ee

    \item Find the downlink power $\{p_k\}$ to achieve the SINR $\gamma^{(j)}$, i.e., to solve the following linear equation set:
    \be
            \frac{ p_k |\qh_{k}^T\bar\qw_k|^2}{\sum\limits_{i=1,i \ne  k}^K p_i |\qh_{k}^T\bar\qw_i|^2 + N_0} = \gamma^{(j)}, k=1,\cdots, K.
    \ee

    \item Update the downlink beamforming vector as $\qw_k = \sqrt{p_k} \bar\qw_k$ and $\qW=[\qw_1, \cdots, \qw_K]$.

    \item Update $\bm\mu$ using the subgradient Euclidean projection method with step size $\alpha_j$:
      \be
            \bm\mu^{(j+1)} = \mathcal{P}_{\mathcal{S}}\{ \bm\mu^{(j)}_k + \alpha_j \mbox{Diag}\{\|\qe_n^T \qW\|^2\}\},
      \ee
      where $\mathcal{S}=\{\bm \mu | \sum_{n=1}^{N_t} \mu_n P_n=1\}$.

      This projection $\mathcal{P}_{\mathcal{S}}$ can be solved efficiently using the bisection search. The detailed projection algorithm is provided in Algorithm 3 of Appendix C.

      \item Regulate the downlink beamforming.
            Update the beamforming vector as follows to satisfy all per-antenna power constraints:
            \be
                \qW(n,:) =   {\qW(n,:)} \sqrt{ \min_n \frac{P_n}{ \|\qe_n^T \qW\|^2}}, \forall n.
            \ee

 \end{enumerate}
\begin{figure}[h]
  \centering
  \subfigure[]{
    \includegraphics[width=0.35\textwidth]{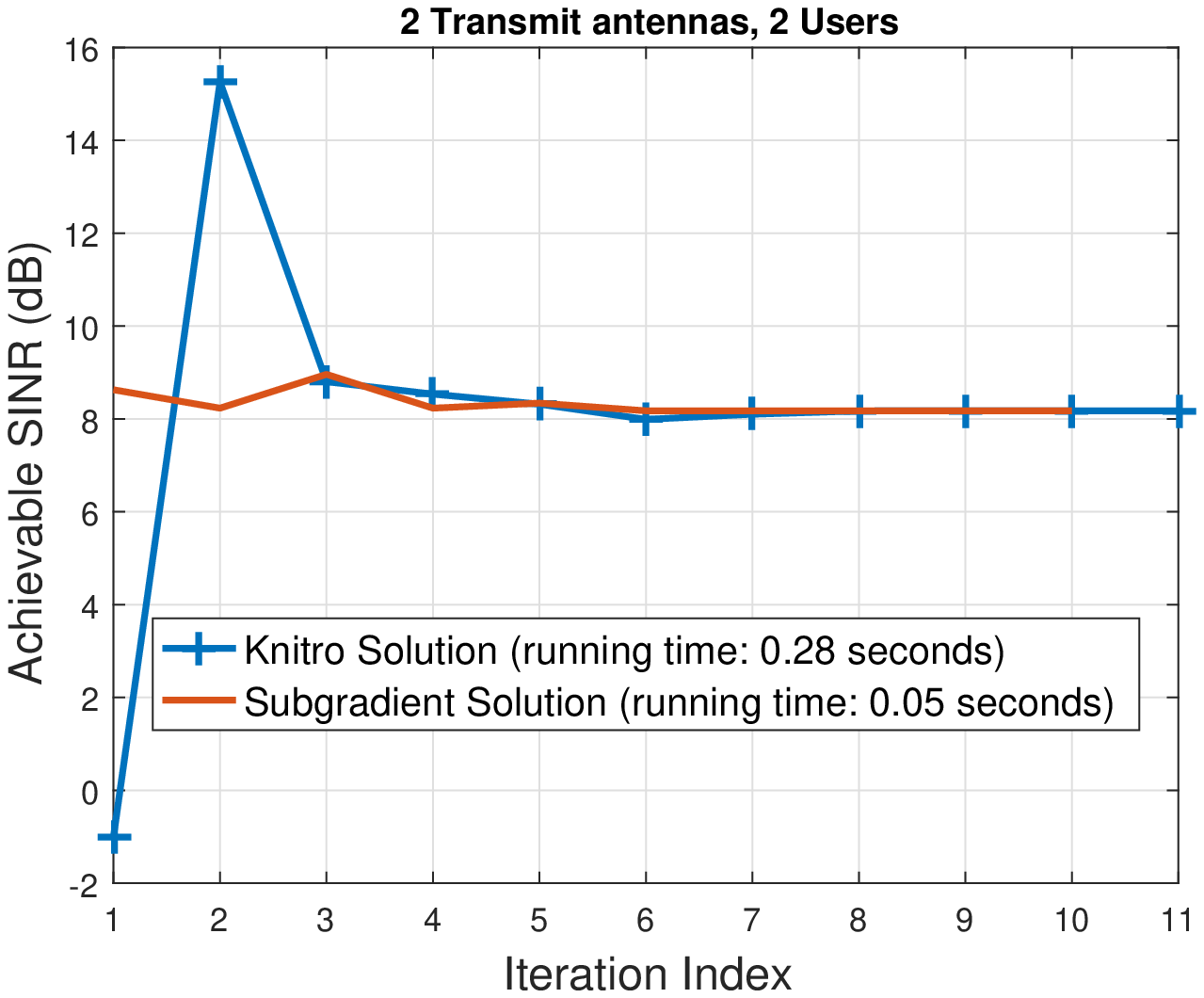}}
      \subfigure[]{
    \includegraphics[width=0.35\textwidth]{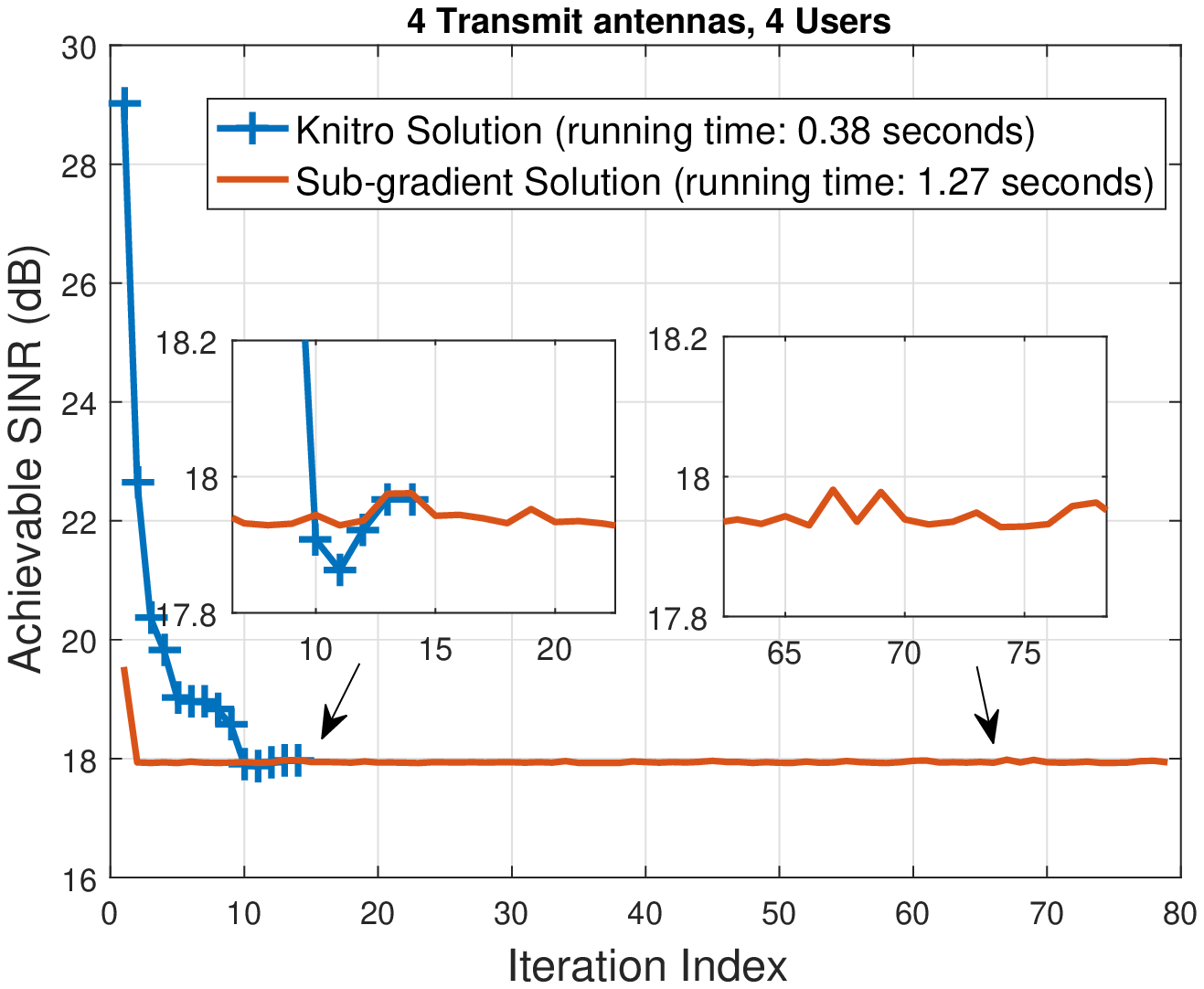}}
  \caption{Comparison of convergence behaviours of the sub-gradient algorithm (Algorithm 2) and the optimal solution using Knitro for two channel instances. The per-antenna power constraint is 10 dB.  (a)$N_t=K=2$ and the sub-gradient algorithm shows faster convergence; (b) $N_t=K=4$ and the sub-gradient algorithm experiences slower convergence and    converges only to the neighbourhood of the optimal solution. }
\label{Fig:convergence}
\end{figure}
{\it Remarks:} The subgradient algorithm exploits the structure of the original problem \textbf{P1}, so it is more efficient than a general nonlinear solver. However, the step size $\alpha_j$ is a critical parameter. We find that $\alpha_j= 0.01 \times 2^{-j} $ gives a satisfactory performance. We observe that in general   Algorithm 2 can solve \textbf{P3} faster than the available numerical solver such as Knitro and achieve close to optimal performance, which can be seen in Fig. \ref{Fig:convergence}(a) and will be verified using simulation results in Section V. However, there is no guarantee that a subgradient algorithm  converges to the exact optimal solution. It  may only converge to the neighbourhood  of the optimal solution, and its convergence may be slow, as seen in Fig. \ref{Fig:convergence}(b). In addition, the subgradient algorithm may not guarantee that the per-antenna power constraints will be satisfied, and that is why Step 8) of Algorithm 2 is necessary to regulate the per-antenna power.

Another important implication of the development of Algorithm 2 is that it provides an efficient way to recover the primal variable, i.e., the downlink beamforming vectors, given various dual variables either $\bm\mu$ and $\bm\lambda$, or only $\bm\mu$. The details will be given in the next section.

\section{The Proposed Deep Learning Structure and Strategies}\label{DL}

In this section, we develop DL based solutions of \textbf{P1} that can achieve better performance-efficiency tradeoff than the currently  available solutions. Instead of learning to optimize the original beamforming matrix $\qW$ directly, we will learn the optimization of the dual variables in \textbf{P2}. This will dramatically reduce the number of variables that need to be learned. In the sequel, we will first introduce a general DL structure that takes the channel {$\qh=[\qh_1^T, \cdots, \qh_k^T]^T$} as the input, and the output is the dual variable(s) in \textbf{P2}. We will also devise a generalized learning solution such that the proposed DL structure can deal with varying number of users and antennas and transmit power without re-training. We will then propose two learning strategies, i.e., one is to learn the dual variables $\bm\mu$ and $\bm\lambda$ with fast recovery of the original beamforming solution, and the other is  to learn only the dual variable $\bm\mu$ with improved learning accuracy, to achieve various tradeoffs.

\subsection{A General DL Structure}
We first show the existence of a neural network that can approximate the solution of  the optimization problem \textbf{P2}. To this end, we define $\bm{\mu}^{opt}$ and $\bm{\lambda}^{opt}$ as two tensors with the optimized dual variables $\bm{\mu}$ and $\bm{\lambda}$, respectively.  The neural network aims to learn the continuous mapping
\begin{equation}\label{mapping function}
\mathcal{F}(\qh,\bm{\mu}^{0})=\{\bm{\mu}^{opt}, \bm{\lambda}^{opt}\},
\end{equation}
where $\bm{\mu}^{0}$ is the initialization set of dual variables and $\mathcal{F}(\cdot,\cdot)$ denotes the continuous mapping process in Algorithm 2 to achieve the stationary point from the input set of channel coefficients together with the initialization set of dual variables.
The following theorem will prove the existence of a feedforward network which imitates the continuous mapping in  \eqref{mapping function}.
\begin{theorem}
For any given accuracy  $\varepsilon>0$, there exists a positive constant $L$ large enough such that a feedforward neural network with $L$ layers can produce similar performance to the mapping process in \eqref{mapping function}, i.e.,
\begin{equation}
\text{sup}_{\qh,\bm{\psi}}||\text{NET}_{L}(\qh,\bm{\psi})-\mathcal{F}(\qh,\bm{\mu}^{0})||_F\leq\varepsilon,
\end{equation}
where $\bm{\psi}$ is the set of the neural network parameters including weights and biases.
\end{theorem}
\begin{proof}
 The result in Theorem 1 can be obtained directly  by applying the universal approximation theorem in \cite{hornik1989multilayer} to the continuous mapping in Algorithm 2.
\end{proof}

Based on results in Theorem 1, next we find solutions through designing the neural networks with the  DL technique.
Similar to our previous work \cite{Wenchao-2019}, we introduce a general DL structure to approximate the mapping function from the channel coefficients to the beamforming solutions, as shown in Fig. \ref{BNN framework}. In addition to the conventional neural network module, the adopted DL structure also introduces a signal processing module based on expert knowledge for beamforming recovery from the key features, such as the dual variables $\bm{\lambda}$ and $\bm{\mu}$ in problem \textbf{P2}. Predicting the beamforming matrix directly may lead to high complexity since the number of the variables in the beamforming matrix depends on both the  number of users $K$ and the number of BS antennas $N_t$. Thus instead of predicting the beamforming matrix directly, we predict some key features (i.e., the dual variables $\bm\mu$ and $\bm\lambda$) whose variables are much less than those  in the beamforming matrix. Then these key features are used to recover the beamforming matrix  in the signal processing module.

 The adopted DL structure takes the convolutional neural network (CNN) architecture as the backbone because the parameter sharing adopted in the CNN can reduce the number of the learned parameters when compared to a fully-connected DNN. Moreover, CNN is well known to be effective for extracting features, which will benefit the generation of the beamforming solution using the channel features.
 The adopted DL structure includes two main modules: the neural network module and the signal processing module \cite{zhang2018hybrid}. Here we give a short description about the two modules, and for more details readers are referred  to \cite{Wenchao-2019}.

\begin{figure}[h]
\centering
\includegraphics[width=0.55\textwidth]{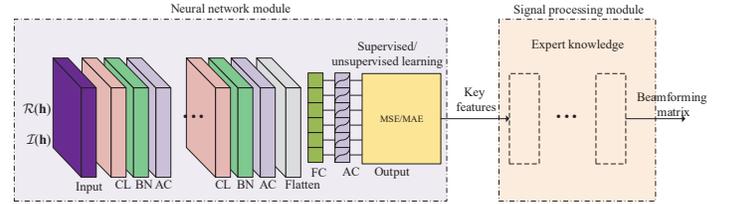}
\caption{A DL-based learning structure for the optimization of downlink beamforming, which includes two main modules: the neural network module and the signal processing module. The neural network module consists of  convolutional (CL) layers, batch normalization (BN) layers, activation (AC) layers, a fully-connected (FC) layer, and so on.  However,  the {{functionalities}} in the signal processing module, as well as the key features input,  are abstract, which are  specified by the expert knowledge.}
\label{BNN framework}
\end{figure}
\subsubsection{Neural Network Module}
The neural network module is a data-driven approach to approximate the mapping function from the complex channels to the key features.  In addition to the input and output layers, the neural network module also includes convolutional (CL) layers, batch normalization (BN) layers, activation (AC) layers, a flatten layer, and a fully-connected (FC) layer.  The input of the neural network module is the complex channel coefficients, which are not supported by the current neural network software. To address this issue, we separate the complex channel vector $\qh=[\qh_1^T,\cdots,\qh_K^T]^T\in\mathbb{C}^{NK\times 1}$ into two components $\mathfrak{R}(\qh)$ and $\mathfrak{I}(\qh)$ {{and form the new input $[\mathfrak{R}(\qh), \mathfrak{I}(\qh)]^T\in \mathbb{R}^{2\times N_tK}$}}, where $\mathfrak{R}(\qh)$ and $\mathfrak{I}(\qh)$ contain the real and imaginary parts of each element in $\qh$, respectively.  Each CL layer consists of many filters which apply convolution operations to the layer input, capture  special patterns and pass  the result to the next layer. The parameters of the filters are shared among different channel coefficients. The main function of the  BN layers is to normalize the output of the CL layers by two trainable parameters, i.e.,  a ``mean'' parameter and a ``standard deviation'' parameter. Besides, the BN layers can reduce the probability of over-fitting and enable a higher learning rate.

AC layers help neural networks extract the useful information and suppress the insignificant points of the input data. The rectified linear unit (ReLU) and sigmoid functions are suitable choices for the last AC layer, since the predicted variables are continuous and positive   numbers. The function of the flatten layer is to change the shape of its input into a vector for the FC layer to interpret. In addition to these functional layers, the loss function, marked `MSE/MAE' on the output layer in Fig. \ref{BNN framework}, is also very important in the introduced DL structure. The mean absolute error (MAE) or  the mean square error (MSE) is used  in the loss function  to update   parameters. The loss function together with the learning rate  determines how to update the  parameters of the neural network module.

\subsubsection{Signal Processing Module}
The neural network module offers universality in learning the key features from data, while
the signal processing module  aims to recover the beamforming matrix from the predicted key features at the output layer.
Different from the neural network module whose model is unknown, the signal processing module utilizes the (partially) known models of the data to recover the beamforming matrix. The learned key features and the {{functionalities}} in the signal processing module are designated according to the expert knowledge.  Note that the expert knowledge is problem-dependent and has no unified form, but what is in common is that the expert knowledge can significantly reduce the number of variables to be predicted compared to the beamforming matrix \cite{Wenchao-2019}. For example, the dual forms of the original problems are the typical expert knowledge for beamforming optimization.  The details of the signal processing module used to recover the beamforming matrix is provided in the next two subsections.

\subsection{To Learn $\bm\lambda$ and $\bm\mu$ and the Recovery Algorithm}
 With the above proposed general   DL structure, we need to decide which features of dual optimization variables in  \textbf{P2} will be learned, and what signal processing function is needed to recover the beamforming matrix. The first option is to learn both  $\bm\lambda$ and $\bm\mu$, so the output has  $K+N_t$ variables. Once they are learned, the following algorithm with steps taken from Algorithm 2 can be used to find a feasible beamforming solution that satisfies the per-antenna power constraints.

\textbf{\underline{Algorithm 4: To recover $\qW$ from $\bm\lambda$ and $\bm\mu$}}
 \begin{enumerate}
    \item Given the learned solution of $\bm\lambda$ and $\bm\mu$,
          Calculate $\beta_k = \frac{1} {\lambda_k \qh_k^T \qG(\bm\lambda, \bm \mu)^{-1}\qh_k^*}$. Then update the achievable SINR in the uplink as
          \be
            \gamma = \min_k  \frac{1}{\beta_k-1}.
          \ee

    \item Find the optimal normalized uplink beamforming as
           \be
            \bar \qw_k = \frac{\qG_k(\bm\lambda, \bm \mu)^{-1}\qh_k^*}{\|\qG_k(\bm\lambda, \bm \mu)^{-1}\qh_k^*\|}, \forall k.
         \ee

    \item Find the downlink power $\{p_k\}$ to achieve the SINR $\gamma$, i.e., to solve the following linear equation set:
    \be
            \frac{ p_k |\qh_{k}^T\bar\qw_k|^2}{\sum\limits_{i=1,i \ne  k}^K p_i |\qh_{k}^T\bar\qw_i|^2 + N_0} = \gamma^{*}, k=1,\cdots, K.
    \ee

    \item Update the downlink beamforming vector as $\qw_k = \sqrt{p_k} \bar\qw_k$ and $\qW=[\qw_1, \cdots, \qw_K]$.

         \item Regulate the downlink beamforming. Update the beamforming vector as follows to satisfy all per-antenna power constraints:
            \be
                {\qW(n,:)} =  {\qW(n,:)} \sqrt{ \min_n \frac{P_n}{ \|\qe_n^T \qW\|^2}}, \forall n.
            \ee

 \end{enumerate}

\subsection{To Learn  $\bm\mu$ Only and the Recovery Algorithm}
 The above learning strategy is straightforward and fast if the learning result is satisfactory, however, the learning accuracy can be much improved if the number of variables is reduced. This motivates us to use the proposed DL structure to learn only the dual variable $\bm\mu$ with output size of $N_t$, which contains $K$ less variables than the above approach that learns both  $\bm\lambda$ and $\bm\mu$. The idea of this approach is that given $\bm\mu$, the optimal $\bm\lambda$ can be efficiently optimized using Algorithm 2, which is more accurate than the   learning approach above. An additional advantage is that the output size does not depend on the number of users, so it can more easily adapt to the varying number of users. Once $\bm\mu$ is learned, the following algorithm with steps taken from Algorithm 2 can be used to derive a feasible beamforming solution to the original problem \textbf{P1}.

 \textbf{\underline{Algorithm 5: To recover $\qW$ from $\bm\mu$}}
 \begin{enumerate}
    \item Given the learned solution $\bm\mu$, call  Algorithm 1  to find the optimal $\bar{\bm\lambda}$.

    \item
          Calculate $\beta_k = \frac{1} {\lambda_k \qh_k^T \qG(\bar{\bm\lambda}, \bm \mu)^{-1}\qh_k^*}$. Then update the achievable SINR in the uplink as
          \be
            \gamma = \min_k  \frac{1}{\beta_k-1}.
          \ee

    \item Find the optimal normalized uplink beamforming as
           \be
            \bar \qw_k = \frac{\qG_k(\bar{\bm\lambda}, \bm \mu)^{-1}\qh_k^*}{\|\qG_k(\bar{\bm\lambda}, \bm \mu)^{-1}\qh_k^*\|}, \forall k.
         \ee

    \item Find the downlink power $\{p_k\}$ to achieve the SINR $\gamma$, i.e., to solve the following linear equation set:
    \be
            \frac{ p_k |\qh_{k}^T\bar\qw_k|^2}{\sum\limits_{i=1,i \ne  k}^K p_i |\qh_{k}^T\bar\qw_i|^2 + N_0} = \gamma^{(i)}, k=1,\cdots, K.
    \ee

    \item Update the downlink beamforming vector as $\qw_k = \sqrt{p_k} \bar\qw_k$ and $\qW=[\qw_1, \cdots, \qw_K]$.

      \item Regulate the downlink beamforming. Update the beamforming vector as follows to satisfy all per-antenna power constraints:
            \be
                {\qW(n,:)} =  {\qW(n,:)} \sqrt{ \min_n \frac{P_n}{ \|\qe_n^T \qW\|^2}}, \forall n.
            \ee

 \end{enumerate}

\subsection{Generalization of the Proposed DL Structure}

In this section, we will generalize the proposed universal DL so that it can adapt to the change of the number of users and antennas.  Although the above DL approaches can achieve satisfactory performance  for beamforming design, applying the DL approaches to practical applications faces the difficulties caused by the dynamic  wireless networks. In other words, when the number of transmit antennas $N_t$ or the number of users $K$ changes, a new model should be trained for prediction. This fact suggests that the applicability of the DL approaches is limited. Transfer learning and training set augmentation are effective ways to improve the generalization. The former transfers an existing model to a new scenario with some additional training and labelling effort \cite{shen2018transfer}, whereas the latter aims to train a large-scale model which adapts to different $N_t$ and $K$ by adding more samples into the training set, so that the training set can cover more possible scenarios.  In this work, we adopt the latter method for simplicity. Without losing generality, we take the DL approach to learning $\bm{\mu}$ only as an example and give more details about the training set augmentation method.

In the  training set augmentation method, we aim to train a large-scale model with $2N^{\prime}_tK^{\prime}$-input and $N_t^{\prime}$-output. In order to make the large-scale model adaptable to different $N_t$ and $K$ values, we generate an augmented training set. Different from the training set whose   samples have the same  $N_t$ and $K$ values, the samples in the augmented training set are diverse, i.e., the numbers of the transmit antennas and the numbers of users in different samples could vary. However, the size of each sample is fixed as $2N^{\prime}_tK^{\prime}$-input and $N_t^{\prime}$-output.  For the cases where $N_t<N_t^{\prime}$ (or $K<K^{\prime}$) , the redundant $N_t^{\prime}-N_t$ rows (or $K-K_0$ columns) of the channel matrix are filled with 0's. Similarly, the redundant $N_t^{\prime}-N_t$ elements of output are set as 0 when $N_t<N_t^{\prime}$. In each sample, we assume each $K\in\{1,2,\cdots,K^{\prime}\}$ is generated with the equal probability of $\frac{1}{K^{\prime}}$ and each $N_t\in \{1,2,\cdots, N^{\prime}_t\}$ is generated with the equal probability of  $\frac{1}{N^{\prime}_t}$. Therefore, the occurrence probabilities of different $K $  values are statistically equal among all samples and so are different  $N_t$ values.   It is suggested that the number of the samples in  the augmented training set for the large-scale model should be 5-10 times as many as that in the training set with fixed $N_t$ and $K$ values. {{However, this approach works only if the number of users or antennas does not exceed the maximum values used in the training set, otherwise re-training will be needed.}}


\section{Performance Evaluation}\label{Performance}
Both simulations and experiments are carried out to evaluate the performance  of the proposed DL enabled beamforming optimization. We assume that all channel entries undergo independent and identically distributed  Rayleigh flat-fading  with zero mean and unit variance {{unless otherwise specified,}} and perfect CSI is available at the BS. All transmit power is normalized by the noise power.

 {{ The training samples (dual variables) are generated by solving the problem \textbf{P2} using Knitro for its stability and efficiency, but can also be generated by solving the problem \textbf{P1}  using the bisection search method at the cost of more computational time during the offline training. }}
 In our simulation, we use 20000 training samples and 5000 testing samples, respectively.   All of proposed DL networks have one input layer, two CL layers, two BN layers,  three AC layers, one flatten layer, one FC layer, and one output layer. Besides, each CL layer has 8 kernels of size $3\times3$ and the first two AC layers adopt the ReLU function. {{Each CL  applies stride 1 and zero padding 1 such that the output  width and height of all CLs remain  the same as those of the input  \cite{CNN}. To be specific, the input size of the first CL is $2\times N_tK\times 1$ and the output size is $2\times N_tK\times 8$. Both the input size and output size of the second CL are $2\times N_tK\times 8$. When parameter sharing is considered, the numbers of parameters in the first and second CL are $3\times3\times1\times8=72$(weights)+8(bias)=80, and  $3\times3\times8\times8+ 8=584$, respectively, with a total of 664. When no  parameter sharing is considered, the   numbers of parameters in the two CLs are
$(2\times N_tK\times8)\times(3\times3\times1+1)=160N_tK$  and $(2\times N_tK\times8)\times(3\times3\times8+1)=1168N_tK$, respectively, with a total of $1328N_tK$. }} Adam optimizer \cite{ba2015adam} is used with the mean squared error based loss function.    We adopt the sigmoid function in the last AC layer.

 We will compare the performance and running time of the following  schemes when possible:
 \begin{enumerate}
    \item The optimal solution to solve P2 using Knitro.
    \item The proposed subgradient algorithm (Algorithm 2)  in Section III.
    \item The proposed  solution based on learned $\bm\lambda$ and $\bm\mu$.
    \item The proposed  solution based on learned $\bm\mu$ only.
    \item ZF Solution \cite{ZF-PA}.
     \begin{enumerate}
    \item
    When $N_t=K$,  pseudo inverse of the channel is the optimal beamforming direction, i.e.,
     \be
     \tilde \qW=\qH^\dag(\qH^T \qH^\dag)^{-1},
     \ee
     and the achievable SINR is $\Gamma_{ZF} = \min_n \frac{P_n}{ \|\qe_n^T  \tilde\qW\|^2}$.
     The overall optimal beamforming matrix is given by $\qW = \sqrt{\Gamma_{ZF}}\tilde \qW$.
    \item  However, when $N_t>K$, the optimal solution relies on solving the following SOCP problem \textbf{P7}, so the associated complexity is high:
        \bea
    \textbf{P7:}&& \max_{\qW, \Gamma} ~~\Gamma \\
    \mbox{s.t.}  &&    |\qh_{k}^T\qw_k|^2 \ge \Gamma, \forall k, \notag \\
     && \qh_{k}^T\qw_j=0, \forall k\ne j,\notag\\
     &&p_n= \|\qe_n^T \qW\|^2 \le P_n, ~~~\forall n. \notag
    \eea
    \item When $N_t<K$, there is no feasible ZF solution.
    \end{enumerate}

    \item   RZF  Solution \cite{RZF}. This is a low-complexity heuristic solution that improves the performance of ZF especially at the low SNR region. The  beamforming direction is given by:
     \be
     \tilde \qW=\qH^\dag(\qH^T \qH^\dag + \alpha \qI_{K\times K})^{-1},
     \ee
     where  $\alpha=\frac{K N_0}{\sum_{n=1}^{N_t} P_n}$     and the   overall  beamforming matrix is given by $\qW = \sqrt{ \min_n \frac{P_n}{ \|\qe_n^T  \tilde\qW\|^2}}\tilde \qW$.
 \end{enumerate}
 For fair comparison, the convergence of all iterative algorithms is achieved when the relative change of the objective function values is below $10^{-8}$.
  All algorithms are  implemented on an Intel i7-7700U CPU with 32 GB RAM using Matlab R2017b. One NVIDIA  Titan Xp GPU is used to train the neural network.

\subsection{Simulation Results}
 We first compare the SINR and running time results for a system with $N_t=K=4$ in Fig. \ref{SINR_versus_power_K4}. In Fig. \ref{SINR_versus_power_K4} (a), we can see that both the proposed subgradient solution and the  solution based on learned $\bm\mu$  can achieve close to optimal solution and outperform the RZF solution and the ZF solution especially at the low  signal to noise (SNR) regime. As the SNR increases, all solutions converge to the optimal solution.  Fig. \ref{SINR_versus_power_K4} (b) shows that both of the proposed learning based solutions can achieve more than an order of magnitude gain in terms of computational time when compared to the optimal algorithm.  The proposed subgradient algorithm is more efficient than the optimal solution using Knitro. ZF and RZF solutions have the lowest possible  complexity because there is no optimization involved.  {{In addition, we compare the robustness of various schemes against channel errors in Fig. \ref{Fig:error}. The channel vectors  are modelled as $\qh_k=\bar{\qh}_k +\sigma\mathbf{e}_k,\forall k$, where $\bar{\qh}_k$ is the imperfect channel estimate, $\mathbf{e}_k\sim\mathcal{CN}(\mathbf{0},\qI_N)$ is the channel error vector and $\sigma^2$ is the variance of  channel estimation error.  As expected, we can see that the channel estimation error causes degradation of the SINR performance for all the solutions. However, the results show that the proposed learning based solutions and the optimal solution are very robust, but the performance loss of the ZF and RZF beamforming is severe.}}

 Next we demonstrate the scalability of the algorithms when $N_t=K$ and the number of users varies from 2 to 10 when $P_n=10$ dB in Fig. \ref{SINR_versus_K}. As can be seen from Fig. \ref{SINR_versus_K} (a), as both the numbers of users and antennas   increase, the achievable SINR  first decreases and then increases.  The performance of the ZF and RZF solutions drops quickly.  As the number of users increases, both learning based solutions significantly outperform the ZF solution and the performance gap is enlarged while their gap to the optimal solution remains constant.   Fig. \ref{SINR_versus_K} (b) shows the complexity performance. The proposed algorithm that learns both $\bm\lambda$ and $\bm\mu$ has a lower complexity. As the number of users increases, e.g., when $K=10$, it can achieve nearly 50-fold gain in terms of computational time  when compared to the optimal algorithm.  The proposed algorithm that learns only   $\bm\mu$   achieves 0.5 dB higher SINR than that learns both $\bm\lambda$ and $\bm\mu$ at the cost of slightly increased time complexity.   {{Next we examine the SINR performance of the system using a  more realistic 3GPP
Spatial Channel Model (3GPP TR 25.996) \cite{SCM} as shown in Fig. \ref{SINR_large_scale}. We consider a scenario of urban micro cells and assume the distances between the BS and the users are between  50 m  and 300 m and distributed uniformly. The total system bandwidth is 20 MHz. Similar trends of the algorithms are observed in Fig. \ref{SINR_large_scale} as those in Fig. \ref{SINR_versus_K} (a), and both learning based solutions still significantly outperform the RZF and the ZF solutions.}}

 We then consider the performance of a system with $N_t=10$ transmit antennas at the BS, and vary the number of users $K$ when $P_n=10$ dB in Fig. \ref{SINR_Nt_10} (a). It is noticed that there is about 1 to 2 dB gap between the learned solutions and the optimal solution, while the ZF solution is almost optimal when $N_t>K$. However, from Fig. \ref{SINR_Nt_10} (b), we can see that the ZF solution has the highest complexity in this case because its solution needs to be optimized via solving the SOCP problem $\textbf{P7}$. The proposed algorithm that learns both $\bm\lambda$ and $\bm\mu$ achieves more than two orders of magnitude gain in terms of computational  complexity when compared to the ZF solution.

 Next we demonstrate the generalization property of our proposed algorithm that learns only $\bm\mu$. We train a  model with $N_t=K=10$ only once, and then use it when $N_t\le 10$ and $K\le 10$ vary.  As shown in Fig. \ref{generality}, it is observed the SINR performances of the optimal solution and the proposed generalization algorithm using the same model not only has the same trend with respect to the number of  users,  but also are   close to each other. More specifically, the achieved SINRs of the two schemes decrease with the increase of the user number when the number of BS antennas  is fixed.  Such observation validates the feasibility of the training set augment method and motivates further research on improving the generalization of the proposed DL-based algorithms. Besides, we find that adding   more antennas can improve the SINR performance because of the spatial gain.

 \begin{figure}
  \centering
  \subfigure[]{
    \label{snr_feasibility_versus_sinr} 
    \includegraphics[width=0.35\textwidth]{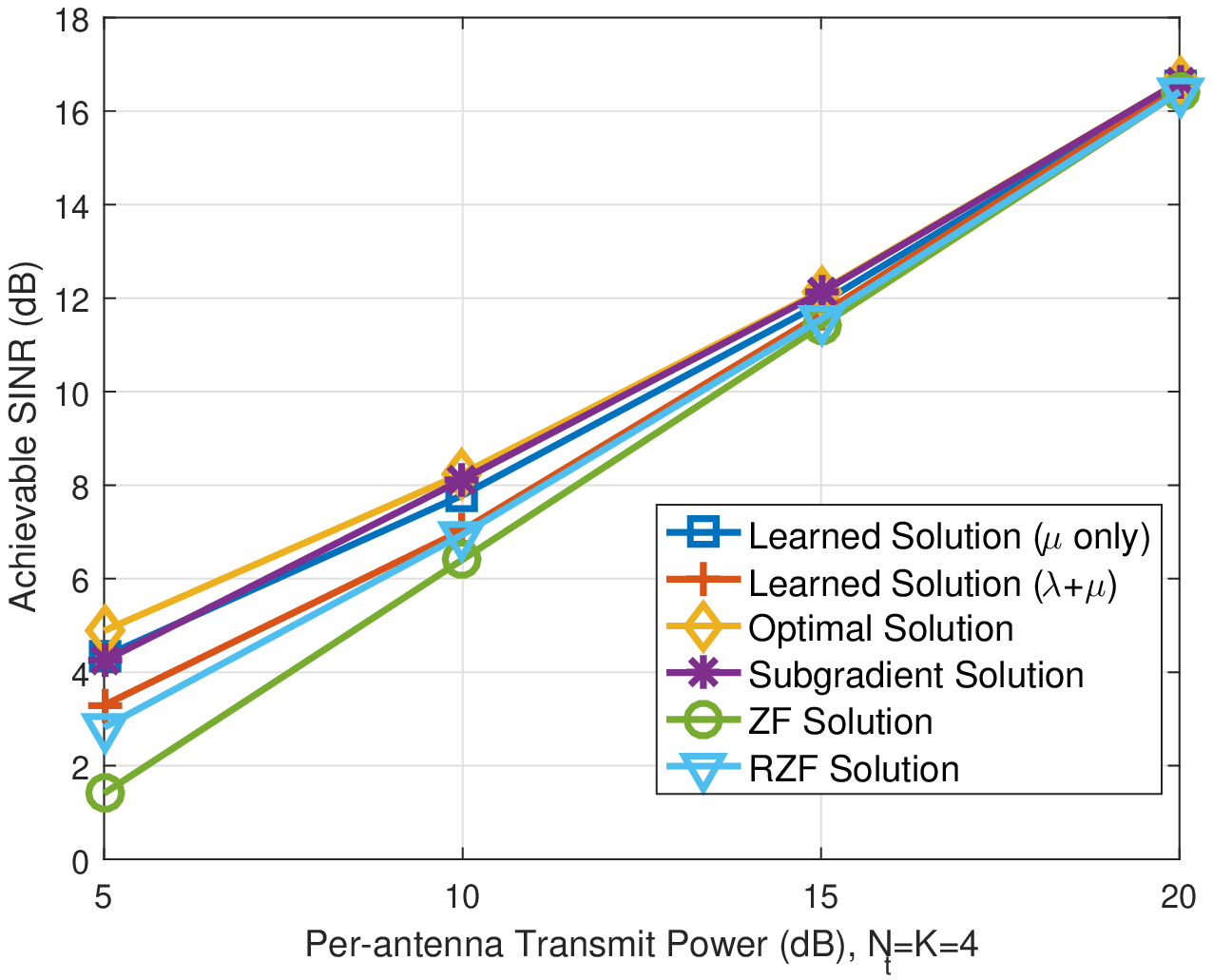}}
  \subfigure[]{
    \label{power_feasibility_versus_sinr} 
    \includegraphics[width=0.35\textwidth]{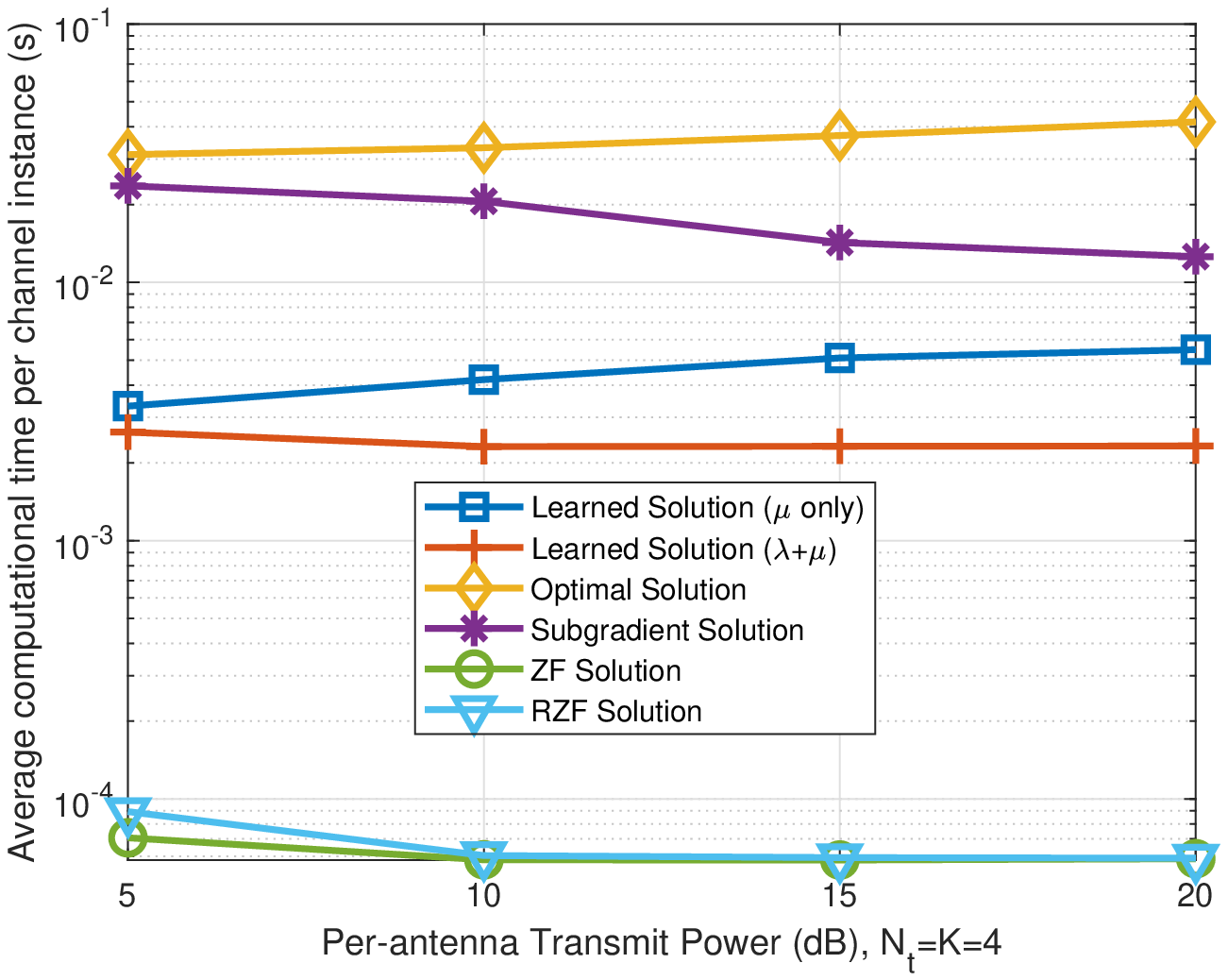}} 
  \caption{The   performance and complexity of a system with $N_t=K=4$  averaged over 5000 samples: (a) minimum   SINR and (b) time consumption per channel realization.}
  \label{SINR_versus_power_K4}
\end{figure}

\begin{figure}[h]
  \centering
    \includegraphics[width=0.35\textwidth]{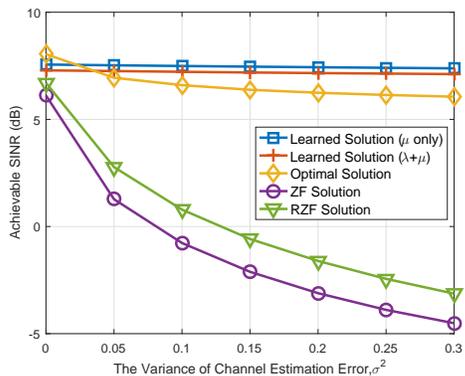}
  \caption{Effect of imperfect CSI on the performance of different schemes for a system with $N_t=K=4$ when $P_n=10$ dB.}
\label{Fig:error}
\end{figure}

 \begin{figure}
  \centering
  \subfigure[]{
    \includegraphics[width=0.35\textwidth]{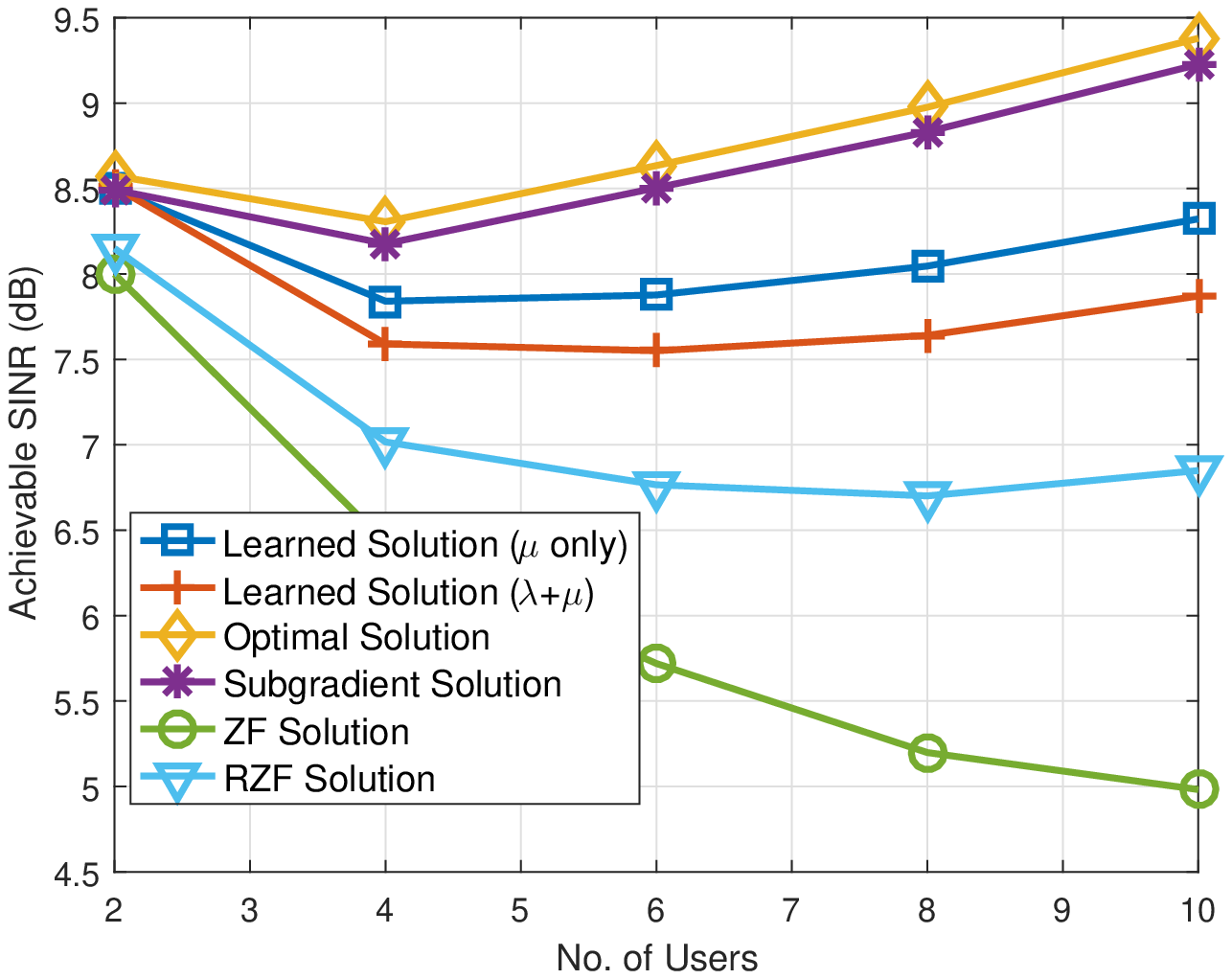}}
  \subfigure[]{
    \includegraphics[width=0.35\textwidth]{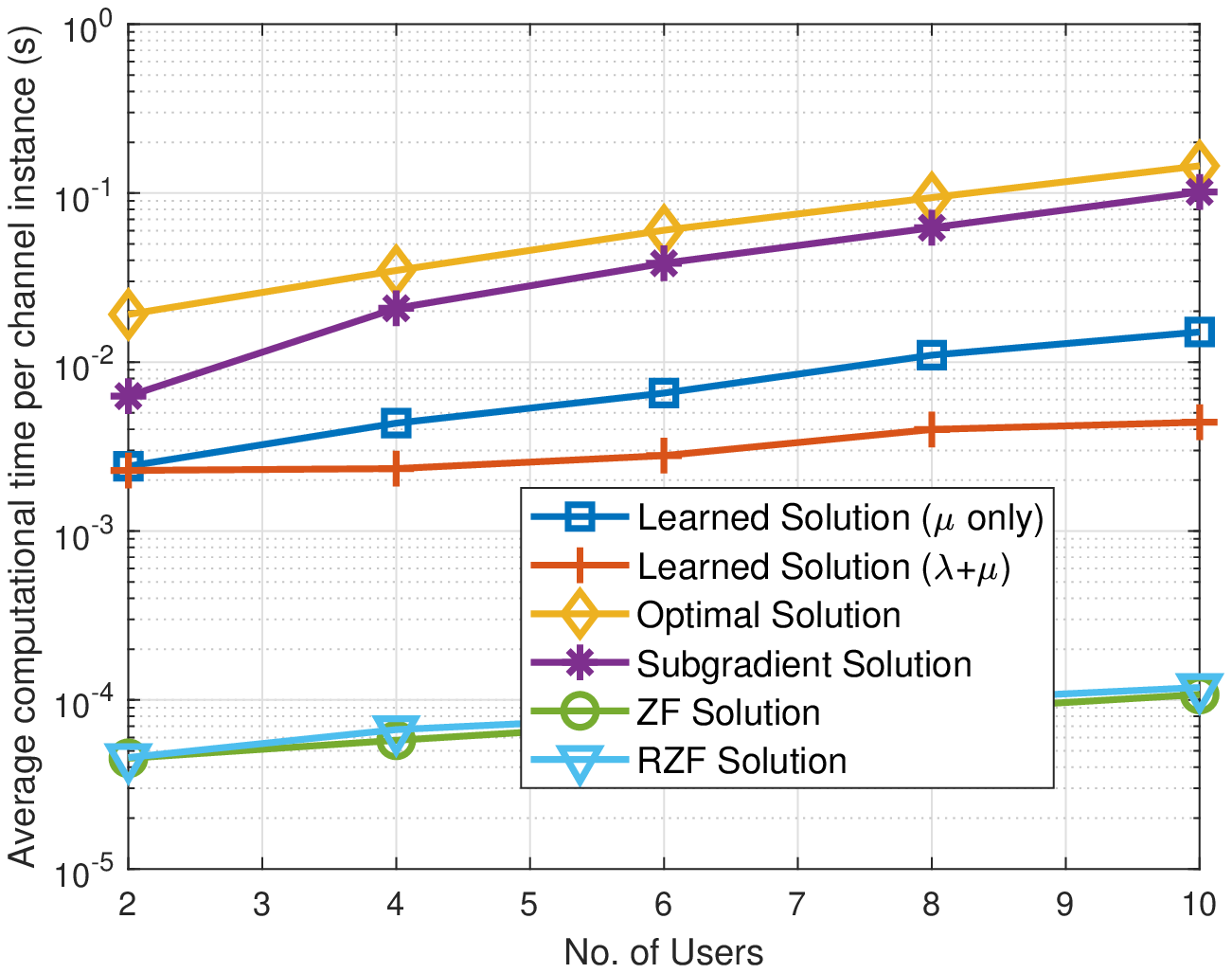}} 
  \caption{The performance and complexity of an $N_t=K$ system when $P_n=10$ dB,  averaged over 5000 samples: (a) minimum SINR and (b) time consumption per channel realization.}
  \label{SINR_versus_K}
\end{figure}

 \begin{figure}
  \centering
    \includegraphics[width=0.35\textwidth]{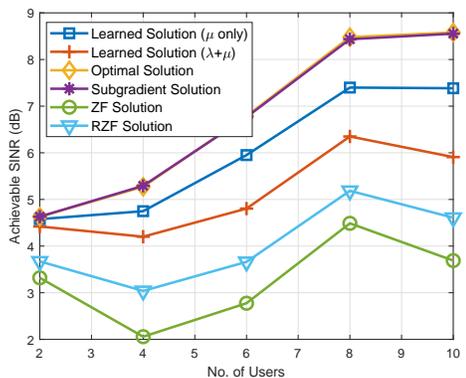}
   \caption{The SINR performance an $N_t=K$ system averaged over 5000 samples using the 3GPP
Spatial Channel Model in an urban micro cell environment when $P_n=30$ dBm.}
  \label{SINR_large_scale}
\end{figure}

  \begin{figure}
  \centering
  \subfigure[]{
    \label{snr_feasibility_versus_sinr} 
    \includegraphics[width=0.35\textwidth]{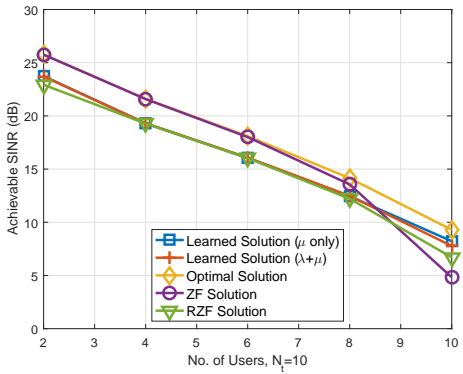}}
  \subfigure[]{
    \label{power_feasibility_versus_sinr} 
    \includegraphics[width=0.35\textwidth]{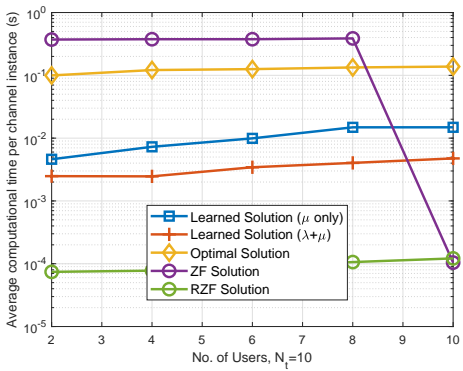}} 
  \caption{The   performance and complexity of a system with $N_t=10$ and varying $K$ when $P_n=10$ dB, over 5000 samples: (a) minimum SINR and (b) time consumption per channel realization.}
  \label{SINR_Nt_10}
\end{figure}

\begin{figure}[h]
\centering
\includegraphics[width=0.35\textwidth]{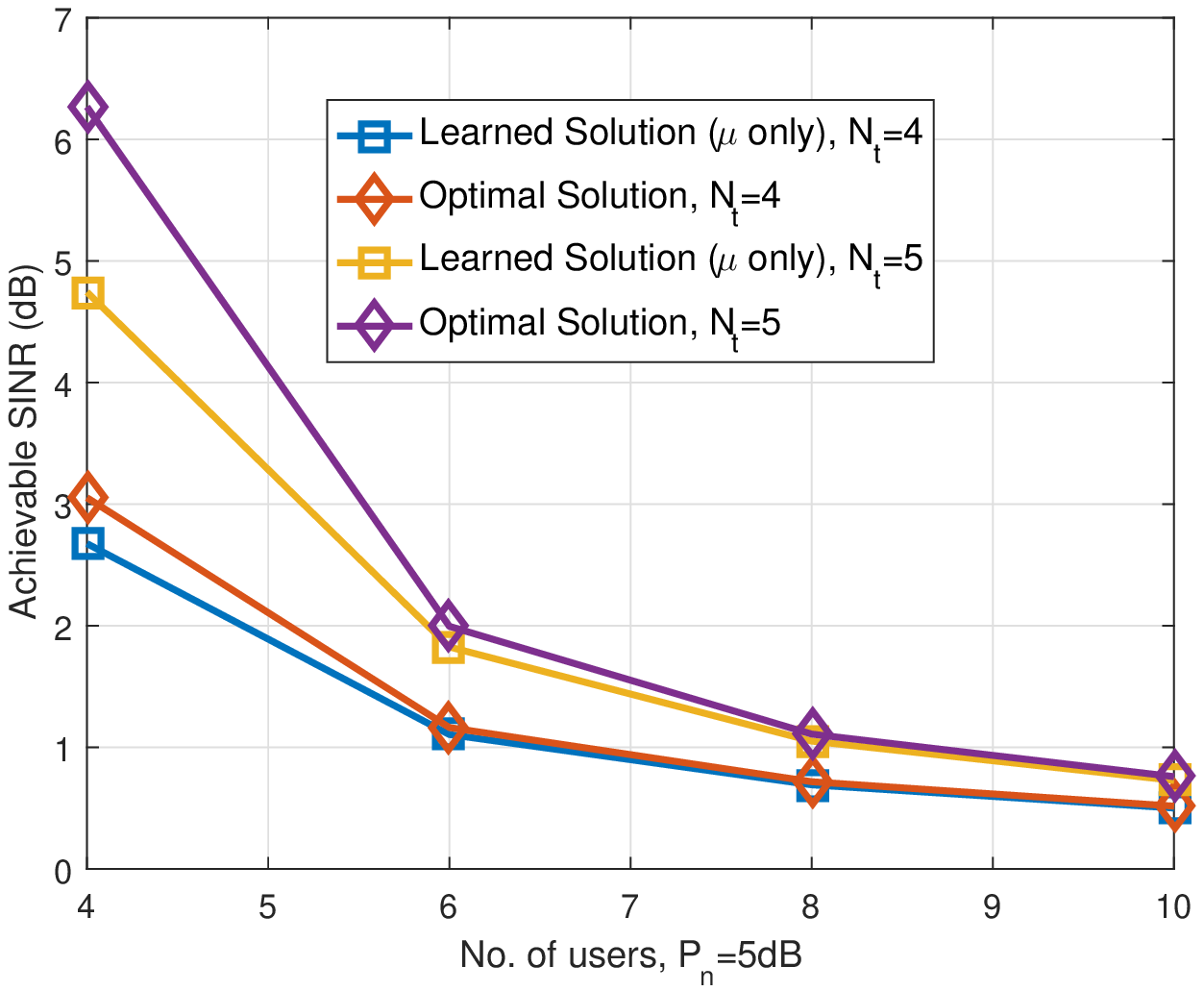}
\caption{The SINR performance with varying $K$ and $N_t$  using the same trained system under $N_t=K=10$. }
\label{generality}
\end{figure}

\subsection{Testbed Results}

To evaluate the proposed learning-based algorithm in a real-world scenario, we have implemented a multi-user beamforming testbed system based on SDR in our   lab environment.
\subsubsection{Testbed Setup}
The multi-user beamforming testbed system is based on the SDR structure, which consists of one PC hosting Matlab, a Gigabit Ethernet switch, four NI's USRP devices as transmitters or receivers and a CDA-2990 Clock Distribution Device. The USRP devices and the Clock Distribution Device for synchronization are illustrated in Fig. \ref{fig hardware}.

We adopt the SDR system since it provides a flexible development environment as well as a practical prototype. The USRP devices are exploited as the radio fronts in the SDR system, which can support different interfacing methods including PCIe and Gigabit Ethernet connections. Besides, the USRP devices can support a wide range of baseband signal processing platforms, including Matlab, Labview and GNU Radio.
The transmitters and receivers are implemented using USRP-2950 devices, which support  the Radio Frequency (RF) range from 50MHz to 2.2GHz \cite{NI_USRP}. For the evaluation purpose, the 900 MHz Industrial, Scientific and Medical (ISM) frequency band is used. The key parameters of the multi-user beamforming system are listed in Table \ref{table system_configuration}.

\begin{figure}[h]
\centering
\includegraphics[width=0.4\textwidth]{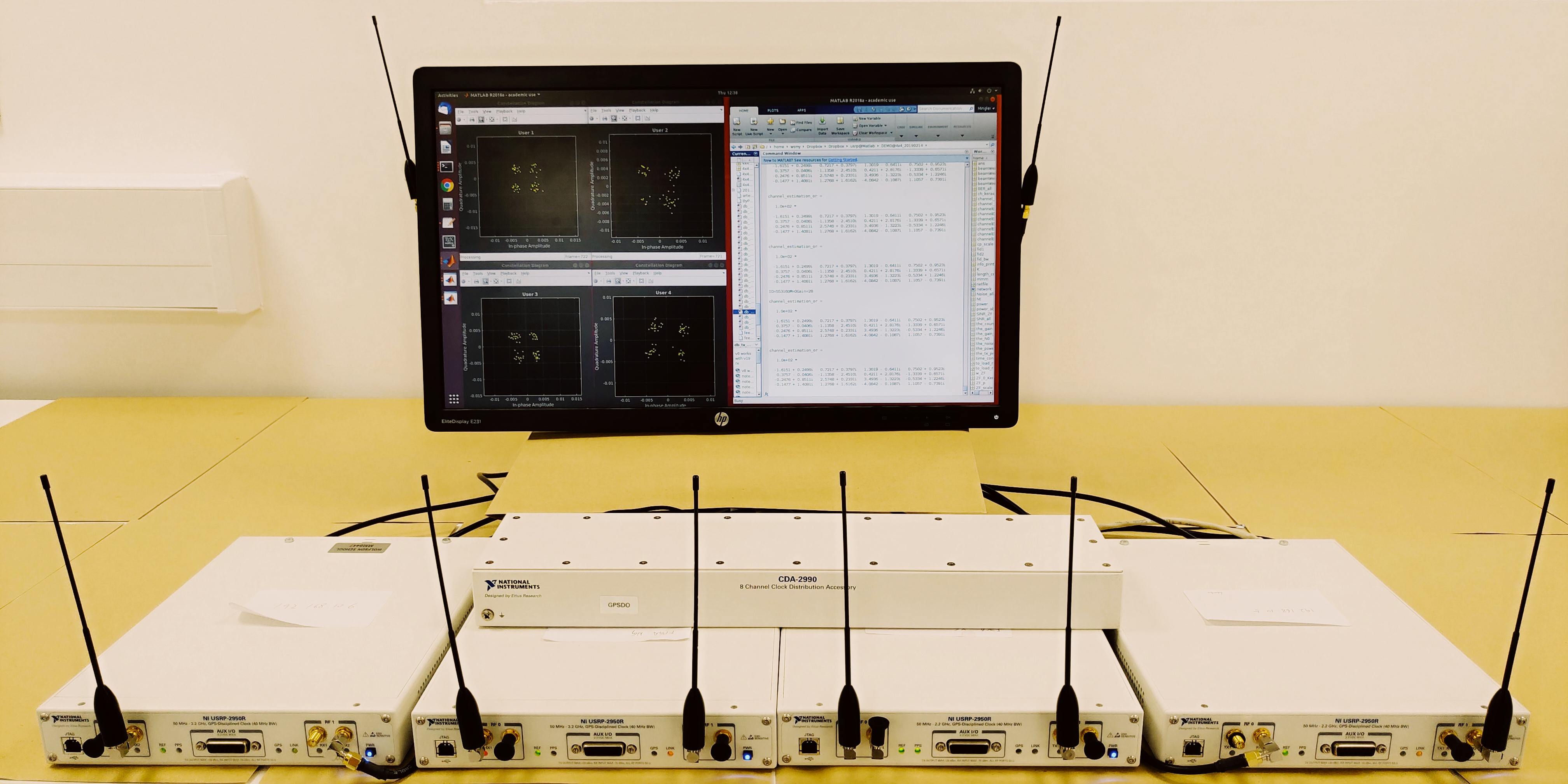}
\caption{The implemented multiuser beamforming testbed system, where two USRPs are combined to make a four-antenna transmitter and two USRPs are used to emulate four single-antenna users.}
\label{fig hardware}
\end{figure}
In the experiment, we consider the scenario consisting of one BS with four transmit antennas and four single-antenna users, i.e., $N_t=K=4$.
We combine  two USRP-2950 devices as a cooperative four-antenna transmitter and employ  two USRP-2950 devices as four individual single-antenna users. All channels on the USRP devices are synchronized using the CDA-2990 Clock Distribution Device. The omnidirectional tri-band SMA-703 antennas are used for both the transmitters and the receivers, while the receiver antennas are extended using RF cables.
Specifically, both static and dynamic channel conditions are examined to evaluate the proposed learning-based beamforming algorithms. For the static channel scenario, the transmitter antennas are placed next to each other with a space of 0.1 m, while the receiver antennas are placed 1.5 m away from the transmitter antennas as well as from each other. For the dynamic scenario, a low-mobility scenario is simulated, where one of the receiving antennas is moving at the speed of 0.6 m/s. Besides, the experiment also exploit different transmitter powers to evaluate the algorithms' performance  in different SNR configurations, where 0 dB of transmit power gain corresponds to a transmit power of $-70$ dBm.
Since the multi-user beamforming system coordinates several USRP devices as transmitters and receivers at the same time, a Gigabit Ethernet switch is used to enable multiple USRP interfacing.

\begin{table}[!h]
\centering
\caption{Testbed system configuration}
\label{table system_configuration}
\begin{tabular}{|l|l|}
\hline
\textbf{Parameters}  & \textbf{Descriptions}                                                                                                                \\ \hline
Clock and PPS source & CDA-2990 10MHz and 1PPS                                                                                                              \\ \hline
Radio Front          & USRP 2950, 50MHz-2.2GHz                                                                                                              \\ \hline
Antennas             & Tri-band SMA-703                                                                                                                     \\ \hline
Modulation           & QPSK                                                                                                                                 \\ \hline
Prefix               & Gold code of length 127                                                                                                              \\ \hline
Baseband Sample Rate & 40 kilosample/second (ksps)                                                                                                                              \\ \hline
Pulse Shaping        & \begin{tabular}[c]{@{}l@{}}Raised Cosine Filter of squared root\\ shape with rolloff factor 0.8 and \\ decimation factor 8\end{tabular} \\ \hline
Channel Estimation   & MMSE Estimator\\ \hline
\end{tabular}
\end{table}


The baseband signal processing modules and the proposed learning-based beamforming algorithms are implemented as Matlab function scripts on a PC with
1 Intel i7-4790 CPU Core, and RAM of 32GB. In the experiment, all users are sharing the same channel and they all use the Quadrature Phase Shift Keying (QPSK) modulation.
The payloads are prefixed with different Gold sequences for each user, which are exploited for both synchronization and channel estimation. Besides, all baseband signals are shaped using a Raised Cosine Filter. During the experiment, each user decodes its own payload and provides channel estimation as feedback to the transmitter. The transmitters and receivers are controlled using different Matlab sessions, while the channel estimation information is exchanged locally on the   PC's cache storage. The beamforming algorithms optimize the beam weight vectors using the aggregated channel estimation information. The transmitter applies the optimized beam weight vectors to generate the signals for each antenna before transmission.

\subsubsection{Experiment Results and Discussions}
To demonstrate the performance of the proposed learning algorithm (based on learned $\bm\lambda$ and $\bm\mu$), three benchmark algorithms are implemented on the multi-user beamforming system, which are the  theoretically  optimal solution, the ZF solution and the RZF solution. Each algorithm is evaluated  under both static and dynamic conditions, and we choose bit error rate (BER) as the performance metric.  In order to generate the BER performance of each solution, a real-time experiment is conducted using the testbed illustrated in Fig. \ref{fig hardware} with different transmitter power. For each transmit power, the BS sends $10^4$ packets each containing 256 QPSK symbols  and the BER is calculated based on the averaged bit error of all packets.

Fig. \ref{experiment} depicts the BER results   in the static and dynamic channel conditions as the transmit power gain varies.
Under the static condition as shown in Fig. \ref{experiment} (a), the proposed learning-based algorithm outperforms the ZF solution and RZF solution across the considered transmit power range. Specifically, the BER performance gain of the learning based algorithm is approximately 4 dB over the ZF solution and 3 dB over the RZF solution in the relatively low transmit SNR regime, and this performance gain reduces as the transmit SNR grows. Compared to the  theoretically  optimal solution, the learning-based algorithm has a close performance in the low transmit SNR regime but becomes inferior for high transmit SNR conditions. This is expected  because under static channel conditions, there is sufficient time to implement the  theoretically  optimal algorithm, therefore it achieves the best performance.
\begin{table*}[!htbp]
\centering
\caption{Typical Time Performance in the   Experiment Scenario}
\label{table Result_Expe_4x4_TimePerformance}
\begin{tabular}{|l|c|c|c|c|c|}
\hline
\textbf{Processing/Solution}                                  & CSI Feedback Period & Learned Solution & Theoretically Optimal Solution & ZF Solution & RZF Solution \\ \hline
\multicolumn{1}{|c|}{\textbf{Typical Time (second)}} & 2x$10^{-2}$                & 5x$10^{-3}$             & 8x$10^{-2}$             & 2x$10^{-4}$        & 2x$10^{-4}$         \\ \hline
\end{tabular}
\end{table*}
However, the algorithms show difference BER performance under the dynamic channel conditions, as depicted in Fig. \ref{experiment} (b). The learning-based algorithm outperforms all benchmark algorithms in the relatively medium to high SNR ranges,  which  corresponds to 0 to 12 dB  in Fig. \ref{experiment} (b). It is worth noting that the learning-based algorithm is superior to the alleged  theoretically  optimal solution under dynamic channel conditions and in particular, the maximum achieved BER performance gain is approximately 1 dB over the  theoretically  optimal solution. This result is expected, and can be explained as follows. The beamforming algorithms require up to date CSI for optimization, but the computational delay of the  theoretically  optimal solution is considerably long, and by the time the  solution is found, the channel would have changed. In other words, the  theoretically optimal beamforming solution is optimized only based on the outdated CSI, and therefore the mismatch  leads to performance degradation, and the  theoretically optimal performance can no longer be guaranteed.
This can be verified by the typical   time-consumption performance for the considered algorithms as illustrated in Table \ref{table Result_Expe_4x4_TimePerformance}.
 This performance degradation becomes worse when the channel conditions are dynamic than that  in the static channel conditions as shown by Fig. \ref{experiment}(a) and Fig. \ref{experiment}(b).
 It is seen from Table \ref{table Result_Expe_4x4_TimePerformance}  that the ZF and RZF solutions require  much less computational time when optimizing the beamforming weights, so the performance of the ZF solution is close to that of the  theoretically  optimal solution (degraded by operating on outdated CSI) in the experiment, and the RZF solution even outperforms the optimal solution.  However, the BER performance of the ZF and RZF solutions is still inferior to that of the proposed learning-based algorithm. It is worth noticing that under both the static and dynamic channel conditions, the precise channel models are not known, so in the experiment, we resort to  the trained neural network based on the small-scale fading for online learning of the beamforming solution. The results in Fig. \ref{experiment} show that the trained network for one channel model generalizes well to cope with different channel conditions and this will greatly reduce the need to re-train the neural network.

   \begin{figure}[]
  \centering
  \subfigure[]{
      \includegraphics[width=0.4\textwidth]{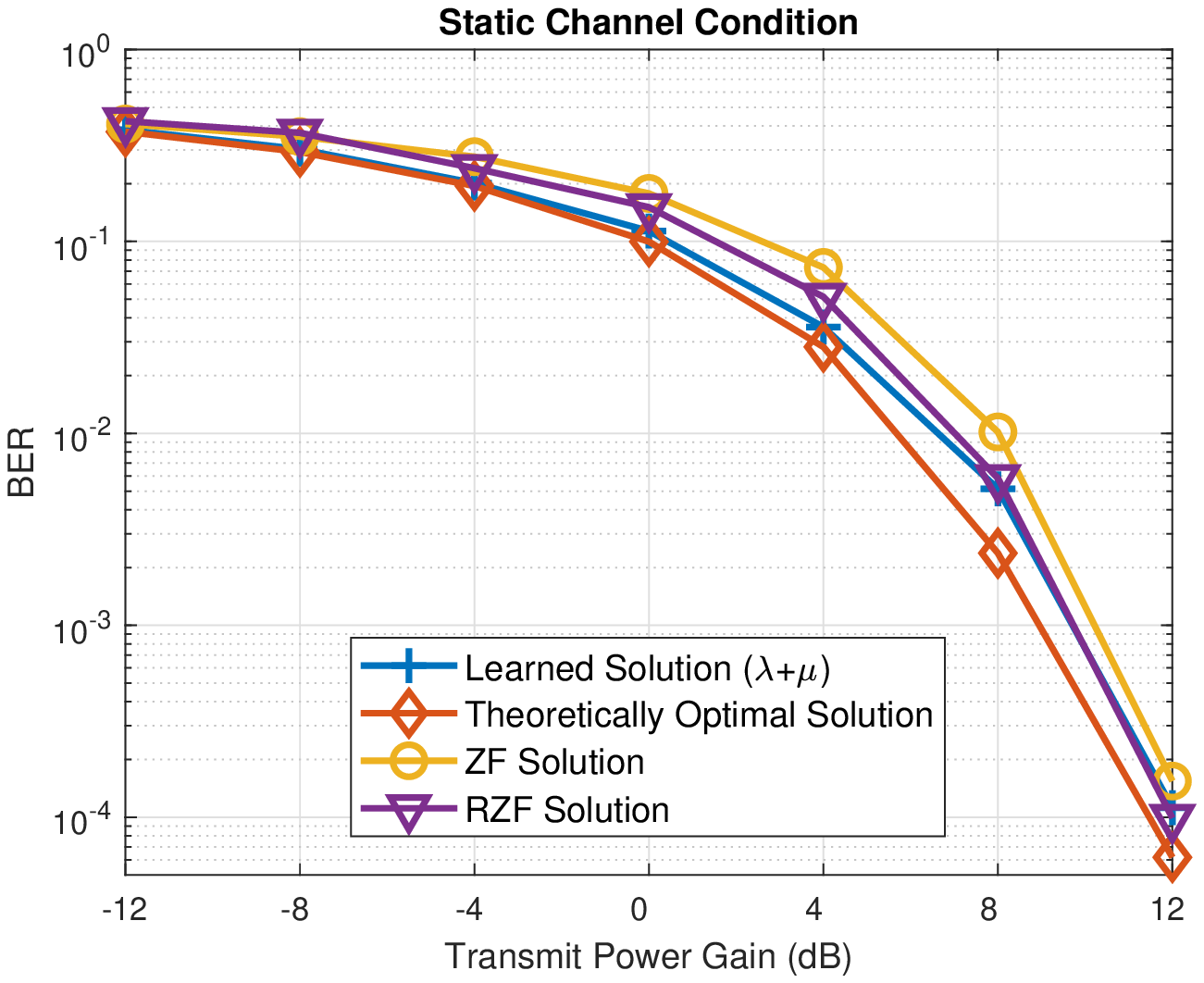}
     }
  \subfigure[]{
     \includegraphics[width=0.4\textwidth]{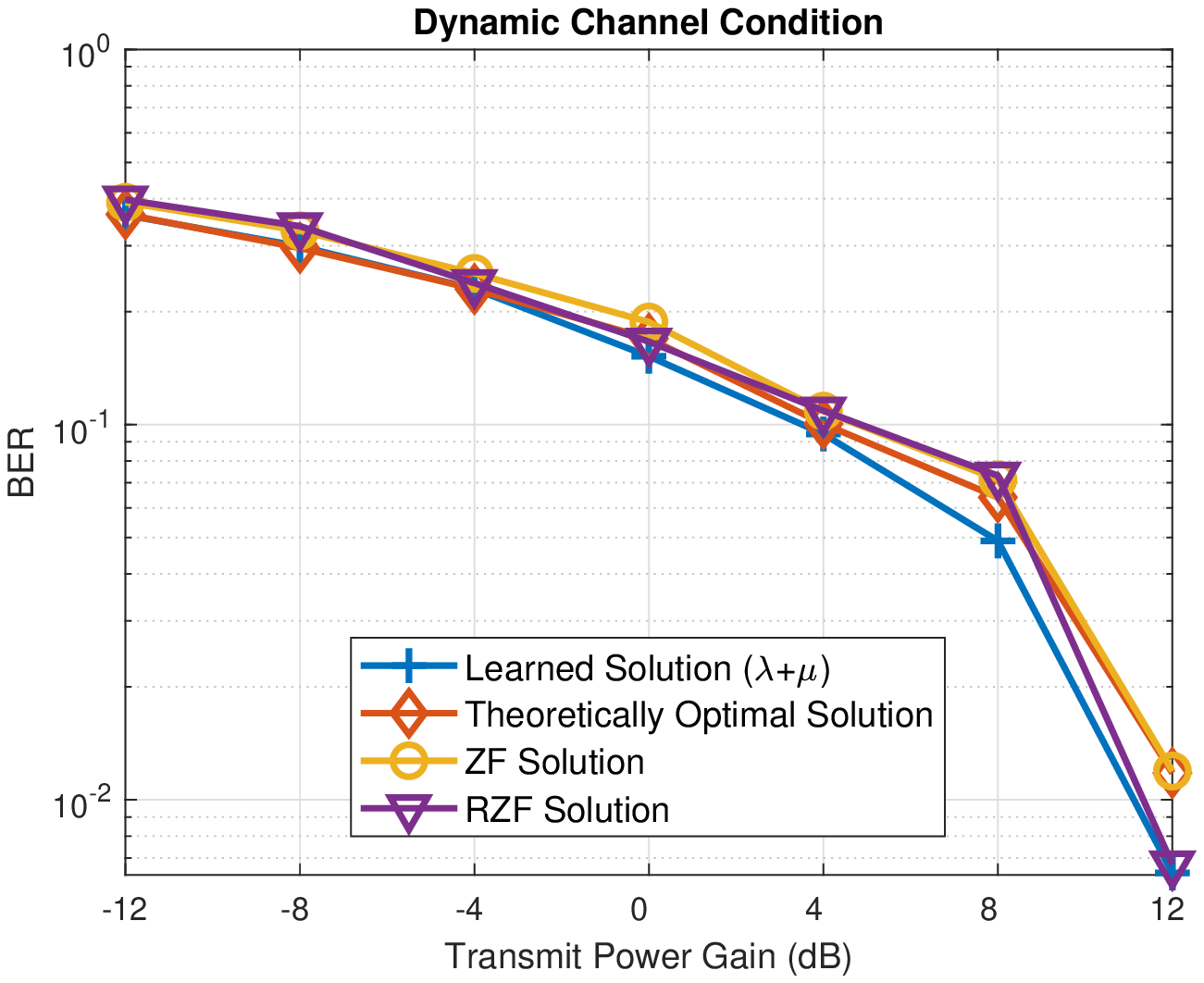}
    }
  \caption{BER performance of the testbed experiments for 4 users, 4 BS antennas scenario: (a) static channel condition, and (b) dynamic channel condition.}
  \label{experiment}
\end{figure}

\section{Conclusions and Future Directions}\label{conclusions}
In this paper, we have developed deep learning enabled solutions for fast optimization of downlink beamforming under the per-antenna power constraints. Our solutions are both model driven and data driven, and are achieved by exploiting the structure of the beamforming problem, learning the dual variables from labelled data and then recovering the original beamforming solutions. Our solutions can naturally adapt to the  varying number of active users in dynamic environments without re-training thus making it more general. The simulation results have shown the superior performance-complexity tradeoff achieved by the proposed solutions, and the results have been further verified by the testbed experiments using software defined radio.

 We would like to point out a few promising future directions.   This paper assumes that perfect CSI is available; however in practice, CSI estimation is never perfect. One future direction would be to investigate a more advanced robust learning framework to mitigate   channel estimation errors or other types of impairments. As a step further, another promising future direction will be to study how to use deep learning to map directly from the pilot signals to the beamformed signals, bypassing the explicit channel estimation step.

 In order to reduce computational complexity of the training process when the channel conditions change, one possible method is to use a wide range of channel realizations during the off-line training phase, in order that the neural network can learn to generalize from a wider range of channel variations. Another  approach is to employ transfer learning  \cite{shen2018transfer}.  The main idea is that knowledge learned from one training task for a given channel condition may be transferred to a similar training task for a different channel condition, and can help train a new model with  additional examples, which is worthy of further study.

\section*{Appendix A. Proof of the convergence and optimality of  Algorithm 1 to Solve \textbf{P4}}

 The proof has two parts. The first part is devoted to the proof of convergence and the second  part addresses the uniqueness and optimality of the fixed point after convergence.

Let us start with $\gamma^{(j)}~(j\ge 1)$ which is achievable for the power vector $\boldsymbol{\lambda}^{(j)}$. It is easily seen    that given $\gamma^{(j)}$, $\mathcal{I} (\bm\lambda^{(j)})$ (user index $k$ is omitted for convenience) is a standard interference function, which satisfies the following properties \cite{Yates95}\cite{Yates98}:
\begin{itemize}
\item[(P1)] $\boldsymbol{\lambda}^{(j)}$ is component-wise monotonically decreasing;
\item[(P2)] If $\boldsymbol{\lambda}\ge \boldsymbol{\lambda}^{'}$, then $\mathcal{I}(\boldsymbol{\lambda})\ge\mathcal{I}(\boldsymbol{\lambda}^{'})$;
\item[(P3)] $\boldsymbol{\lambda}^{(j)}$, for all $j$, are all feasible solutions given the SINR constraint $\gamma^{(j)}$.
\end{itemize}

Assume that at the $j$-th iteration, the dual variable is $\bm\lambda^{j}$ and the achievable SINR is $\gamma^{(j)}$. Then at the $(j+1)$-th iteration, according to (P1), $\bar{\lambda}_k^{(j+1)}\le \lambda_k^{(j)}~\forall k$, and as such $\eta\ge 1$ and $\bar{\lambda}_k^{(j+1)}\le \lambda_k^{(j+1)}~\forall k$ in Step 4). According to P2, in Step 5)   we have the SINR result ${\gamma}_k(\boldsymbol{\lambda}^{(j+1)})>{\gamma}_k(\bar{\boldsymbol{\lambda}}^{(j+1)})$. Then, according to (P3), $ {\gamma}_k
(\bar{\boldsymbol{\lambda}}^{(j+1)})\ge\gamma^{(j)}~\forall k$, and therefore $\gamma^{(j+1)}=\min_k{\gamma}_k (\boldsymbol{\lambda}^{(j+1)})\ge\min_k {\gamma}_k(\bar{\boldsymbol{\lambda}}^{(j+1)})\ge\gamma^{(i)}$, i.e., the balanced SINR $\gamma^{(j)}$ is  increasing as the iteration goes. Since $\gamma^{(j)}$ is upper bounded, the algorithm converges to a fixed point $\boldsymbol{\lambda}^{(\infty)}$. Next, we prove that the fixed point is also optimal.

We see that $\lambda_k^{(\infty)}$ satisfies the following fixed-point equation:
\begin{equation}
\lambda_k^{(\infty)} = \gamma^{(\infty)}  \bar{I}_k(\boldsymbol{\lambda}^{(\infty)})~\forall k.
\end{equation}
and it satisfies the total virtual uplink power is $\sum_k \lambda_k^{(\infty)} = \frac{1}{N_0}$. Clearly, the total uplink transmit power is a monotonic non-decreasing function of the SINR constraint. This implies that there is no solution $\boldsymbol{\lambda}^*$  which provides a strictly higher SINR $\gamma^*>\gamma^{(\infty)}$ but still maintains the power constraint $\sum_k \lambda_k^{(\infty)} = \frac{1}{N_0}$.
\endproof

\section*{Appendix B. Proof that $f(\bm\mu)$ of \textbf{P4}  is a concave function  in  $\bm \mu$.}
{
\begin{proof}
First note that Algorithm 1 to solve \textbf{P4} belongs to a  fixed-point iteration, which means a solution $\{\Gamma, \lambda\}$ that satisfies the first two constraints \eqref{eqn:P41} and \eqref{eqn:P42} with equality ensuring an optimal solution. This indicates there is no local optimum, and the gap between \textbf{P4} and its dual problem is zero. Then it suffices to prove that the objective function of the dual problem of \textbf{P4}  is concave in $\bm\mu$.

By using (11) of [23], we can rewrite \textbf{P4} as
\bea
     \textbf{P4':}&&  f(\bm\mu)=\max_{\Gamma, \bm \lambda} ~~ \Gamma \notag\\
    \mbox{s.t.} && \sum_{i=1}^K \lambda_i \qh_i^*\qh_i^T +\mbox{Diag}(\bm \mu) -\left(1+\frac{1}{\Gamma}\right)\lambda_k \qh_k^* \qh_k^T\succeq \qzero, \forall k,\notag \\
                    && \sum_{k=1}^K \lambda_k N_0 = 1,\notag\\
                && \bm \lambda\ge \qzero.
\eea

Its Lagrangian   is
\bea
   && L_{\bm \mu}(\Gamma, \bm\lambda, a, \qb, \{\qC_k\}) =   \Gamma +a\left( \sum_{k=1}^K \lambda_k N_0 -1\right) + \qb^T \bm\lambda + \\
   && \sum_{k=1}^K \tr\left( \left( \sum_{i=1}^K \lambda_i \qh_i^*\qh_i^T +\mbox{Diag}(\bm \mu) -\left(1+\frac{1}{\Gamma}\right)\lambda_k \qh_k^* \qh_k^T \right) \qC_k)\right),\notag
\eea
where $a, \qb, \{\qC_k\}$ are dual variables. Note that it is derived based on the maximization   rather than the commonly used minimization of an objective function .

The dual objective function  is expressed as $G_{\bm \mu}(a, \qb, \{\qC_k\}) =   \min_{\Gamma, \bm\lambda} L_{\bm \mu}(\Gamma, \bm\lambda, a, \qb, \{\qC_k\})$ which is to be minimized over $(a, \qb, \{\qC_k\})$ and  only contains a linear term of $\sum_{k=1}^K \tr\left(\mbox{Diag}(\bm \mu) \qC_k\right)$ about $\bm\mu$, and the constraints of the dual problem (although not derived here) do not involve $\bm\mu$. Therefore the dual objective function $\min G_{\bm \mu}(a, \qb, \{\qC_k\})$ is a  point-wise minimum of a family of affine functions about $\bm\mu$ and as a result concave \cite[Sec.3.2.2]{convex}, so is $f(\bm\mu)$. This completes the proof.
\end{proof}
}

\section*{Appendix C. To find the subgradient Euclidean projection in  Algorithm 2}
 The Euclidean projection is needed when the  update of $\bm\mu$ based on the subgradient in \textbf{Algorithm 2} does not fall into the feasible set $\mathcal{S}$. It needs to solve the following optimization problem:
 \be
\textbf{P5:}~~\min_{\bm \nu} \|\bm\nu -\bm\mu\|^2 ~~ \mbox{s.t.} \sum_{n=1}^{N_t}   \nu_n P_n=1, \bm\nu\ge 0,
 \ee
 where $\bm\mu =\bm\mu^{(i)}_k + \alpha_i \mbox{Diag}\{\|\qe_n^T \qW\|^2\}$. Although \textbf{P5} is a convex problem and can be solved by a standard numerical algorithm, below we derive its analytical property and  propose a more efficient bisection algorithm to solve it.

 Its   Lagrangian can be expressed as
 \bea
    L = \|\bm\nu -\bm\mu\|^2 + x (\sum \nu_n P_n-1) - \sum_{n} y_n \nu_n,
 \eea
  where $x$ and $y_n\ge 0$ are dual variables.

 Setting its first-order derivative to be zero leads to
 \be
    \nu_n = \frac{2\mu_n + y_n  - x P_n}{2} = \max\left(\frac{2\mu_n  - x P_n}{2}, 0\right).
 \ee

 Substitute it to $\sum_{n=1}^{N_t}   \sum \nu_n P_n=1$ and we get
 \be\label{eqn:bisection:x}
    \sum_{n=1}^{N_t}   \max\left(\frac{2\mu_n  - x P_n}{2}, 0\right) P_n =1.
 \ee
 Therefore the remaining task is to find $x$ that satisfies \eqref{eqn:bisection:x}. Obviously the left hand side of \eqref{eqn:bisection:x} is monotonic in $x$, so we propose the following bisection method to find the optimal $x$.

 \textbf{\underline{ Algorithm 3 to Solve \textbf{P5}:}}
 \begin{enumerate}
    \item Set the upper and lower bounds of $x$ as $x^U$ and $x^L$. Repeat the following steps until convergence.

    \item Calculate $x^t= \frac{x^U+x^L}{2}$.

    \item If $\sum_{n=1}^{N_t}   \max\left(\frac{2\mu_n  - x^t P_n}{2}, 0\right) P_n >1$, $x^L = x^t$; otherwise $x^U=x^t$.

    \end{enumerate}

\end{document}